\documentclass[a4paper, twoside]{amsart}
\usepackage{mathrsfs, tabularx}
\usepackage {amsthm, epsfig, amsmath, amssymb, amscd, float, latexsym}

\theoremstyle{plain}
\newtheorem{theorem}{Theorem}[section]
\newtheorem{proposition}[theorem]{Proposition}

\newtheorem{lemma}[theorem]{Lemma}
\theoremstyle{definition}

\newtheorem{example}[theorem]{Example}

\newtheorem{question}[theorem]{Question}

\newcommand{\kerOp}[0]{\operatorname{ker}}
\newcommand {\Set}[1] {\mathbb{#1}}
\newcommand{\setR}[0]{\Set{R}}

\newcommand{\setC}[0]{\Set{C}}

\newcommand{\slaz}[0]{{\setminus\{0\}}}

\newcommand{\vq}[0]{{\mbox{\boldmath $\xi$}}}

\newcommand{\cA}[0]{{\mathscr{A}}}
\newcommand{\cB}[0]{{\mathscr{B}}}
\newcommand{\cC}[0]{{\mathscr{C}}}
\newcommand{\cD}[0]{{\mathscr{D}}}

\newcommand{\cK}[0]{{\mathscr{K}}}

\newcommand{\cG}[0]{{\mathscr{G}}}

\newcommand {\proofBox}[0]{\hfill $\Box$ }

\newcommand {\proofread}[1]{ }

\newcommand {\IMAG}[0]{ \operatorname{Im}}
\newcommand {\REAL}[0]{ \operatorname{Re}}

\newcommand{\pd}[2]{\frac{\partial #1}{\partial #2}}

\newcommand{\vfield}[1]{{\mathfrak X}( #1)}

\newcommand{\kappaI}[0]{^{(1)}\!\kappa}
\newcommand{\kappaII}[0]{^{(2)}\!\kappa}
\newcommand{\kappaIII}[0]{^{(3)}\!\kappa}

\newcounter{saveenum}

\title{Determination of electromagnetic medium from the Fresnel surface}
\author[Dahl]{Matias F. Dahl}
\address{
Matias Dahl,
Aalto University,
Mathematics,
P.O. Box 11100,
FI-00076 Aalto,
Finland
}

\urladdr{
http://www.math.tkk.fi/\textasciitilde{}fdahl/}


\date{\today}

\begin{document}
\begin{abstract}
  We study Maxwell's equations on a $4$-manifold where the
  electromagnetic medium is described by an antisymmetric $2\choose
  2$-tensor $\kappa$.  In this setting, the Tamm-Rubilar tensor
  density determines a polynomial surface of fourth order  in each
  cotangent space. This surface is called the Fresnel surface and acts
  as a generalisation of the light-cone determined by a Lorentz
  metric; the Fresnel surface parameterises electromagnetic wave-speed
  as a function of direction.  
  Favaro and Bergamin have recently proven that if $\kappa$ has only a
 principal part and if the Fresnel surface of $\kappa$ coincides with
 the light cone for a Lorentz metric $g$, then $\kappa$ is
 proportional to the Hodge star operator of $g$. 
That is, under additional assumptions, the Fresnel surface of $\kappa$
determines the conformal class of $\kappa$.
The purpose of this paper is twofold. First, we provide
 a new proof of this result using  Gr\"obner bases.  
Second, we describe a number of cases where
 the Fresnel surface does not determine the conformal
 class of the original $2\choose 2$-tensor $\kappa$.
For example, if $\kappa$ is invertible we show that
$\kappa$ and $\kappa^{-1}$ have the same
Fresnel surfaces.
\end{abstract}

\maketitle

\section{Introduction}
\label{sec:Intro}
The purpose of this work is to study properties of propagating
electromagnetic fields in linear medium. We will work in a
relativistic setting where Maxwell's equations are written on a
$4$-manifold and the electromagnetic medium is represented by an
antisymmetric $2\choose 2$-tensor $\kappa$. Pointwise, such medium is determined by 
 $36$ parameters. 
 To understand the propagation of an electromagnetic wave in this setting,  the key object 
 is the \emph{Fresnel surface}, which can be seen a
 generalisation of the light-cone
\cite{Rubilar2002, Obu:2003,  PunziEtAl:2009}.
 For a Lorentz metric, the light-cone is always a polynomial surface of second order
in each cotangent space. 
The Fresnel surface, in turn, is a polynomial
 surface of fourth order. For example, the
 Fresnel surface can be the union of two light-cones. This allows the
 Fresnel surface to model propagation also in birefringent medium. That
 is, in medium where differently polarised electromagnetic waves can
 propagate with different wave speeds.

 The Fresnel surface is determined by the \emph{Tamm-Rubilar tensor
   density} which is a symmetric $4\choose 0$-tensor density, which,
 in turn, is determined by the medium $2\choose 2$-tensor $\kappa$.
 This dependence is illustrated in the diagram below:
\begin{eqnarray*}
\label{ill:twoArrow}
\quad\mbox{Medium}\,\, \kappa\quad\to\quad \mbox{Tamm-Rubilar tensor density}\quad\to\quad \mbox{Fresnel surface}.
\end{eqnarray*}
In Lorentz geometry, we know the the light cone of a Lorentz metric
$g$ uniquely determine $g$ up to a conformal factor
\cite{Ehrlich:1991}.  In this work we will study the analogue relation
between a general electromagnetic medium tensor $\kappa$ and its
Fresnel surface; Can one reconstruct an electromagnetic medium from
its Fresnel surface?  In general, a unique reconstruction is not
possible.  For example, the Fresnel surface is invariant under a
conformal change in the medium. Hence the Fresnel surface can, at
best, determine $\kappa$ up to a conformal factor.  One would then
like to understand the following question:

\begin{question}
\label{mainProblem}
Under what assumptions does the Fresnel surface at a point $p\in N$
determine the electromagnetic medium $\kappa\vert_p$ up to a conformal
factor?
\end{question}


In terms of physics, Question \ref{mainProblem} asks when we can
reconstruct $\kappa\vert_p$ (up to a conformal factor) using only
wavespeed information about the medium at $p$. A proper understanding
of this question is not only of theoretical interest, but also of
interest in engineering applications like electromagnetic tomography.
Question \ref{mainProblem} is also similar is spirit to a question in
general relativity, where one would like to understand when the the
conformal class of a Lorentz metric can be determined by the five
dimensional manifold of null-geodesics \cite{Low:2005}.

Favaro and Bergamin have recently proven the following result of
positive nature \cite{FavaroBergamin:2011}: If $\kappa$ has only a
principal part and if the Fresnel surface of $\kappa$ coincides with the
light cone for a Lorentz metric $g$, then $\kappa$ is proportional to
the Hodge star operator of $g$.  That is, in a restricted class of medium,
the Fresnel surface of $\kappa$ determines the conformal class of
$\kappa$.  An important corollary 
is the following: If $\kappa$ has only a principal part and its
Fresnel surface coincides with the light cone for a Lorentz metric,
then $\kappa$ satisfies the \emph{closure condition} $\kappa^2=-f
\operatorname{Id}$ for a function $f\colon N\to (0,\infty)$.
This resolves a conjecture on whether the closure condition
characterises non-birefringent medium in skewon-free medium
\cite{ObuFukRub:00, Obu:2003}.  That the closure condition is sufficient, was
already proven in \cite{ObukhovHehl:1999, ObuFukRub:00}, but before
\cite{FavaroBergamin:2011} sufficiency was only known under additional
assumptions; a proof assuming that $\cC=0$ (see Section
\ref{sec:ABCDtransRules} for the definition of $\cC$ in terms on
$\kappa$) is given in \cite{ObuFukRub:00}, and a proof in a special
class of non-linear medium is given in \cite{ObuRub:2002}. 
For additional positive results to Question \ref{mainProblem}, see
\cite{LamHeh:2004, Itin:2005, Schuller:2010, FavaroBergamin:2011}.

The main contribution of this paper is twofold. First, we give a new
proof of the result quoted above from \cite{FavaroBergamin:2011}.
This is formulated as implication \ref{coIII} $\Rightarrow$ \ref{coII}
in Theorem \ref{thm:mainResult}.  While the original proof in
\cite{FavaroBergamin:2011} relies on the classification of skewon-free
$2\choose 2$-tensors into 23 normal forms by Schuller, Witte, and
Wohlfarth \cite{Schuller:2010}, we will use Gr\"obner bases to prove
Theorem \ref{thm:mainResult}. Essentially, Gr\"obner bases is a
computer algebra technique for simplifying a system polynomial
equations without changing the solution set. See Appendix
\ref{app:Groebner}.

The second contribution of this paper is given in Section
\ref{sec:uni} which contains a number of cases, where the Fresnel
surface does not determine $\kappa$.  In Theorem \ref{thm:FkappaInvB}
\ref{thm:FkappaInvB:iv} we show that if $\kappa$ is invertible, then
$\kappa$ and $\kappa^{-1}$ have the same Fresnel surfaces.  Also, in
Example \ref{ex:complexExample} we construct a $\kappa$ with complex
coefficients on $\setR^4$.  At each $p\in \setR^4$, this medium is
determined by one arbitrary complex number, and hence the medium can
depend on both time and space. However, at each point, the Fresnel
surface of $\kappa$ coincides with the usual light cone of the flat
Minkowski metric $g=\operatorname{diag}(-1,1,1,1)$.

The paper is organised as follows.  In Section \ref{mainSec} we review
Maxwell's equations and linear electromagnetic medium on a
$4$-manifold.  In Section \ref{sec:GOS} we describe how the
Tamm-Rubilar tensor density and Fresnel surface is related to wave
propagation.  To derive these objects we use the approach of geometric
optics.
As described in Section \ref{sec:GOS}, this can be seen as a step
towards a relativistic theory of electromagnetic Gaussian beams (if
such a theory exists).  In general, Gaussian beams is an asymptotic
technique for studying propagation of waves in hyperbolic systems.
These solutions behave as wave packets; at each time instant, the
entire energy of the solution is concentrated around one point in
space. When time moves forward, the beam propagates along a curve, but
always retains its shape of a Gaussian bell curve. Electromagnetic
Gaussian beams are also known as quasi-photons \cite{Kachalov:2002,
  Kachalov:2004, Kachalov:2005, DahlPIER:2006}.  For the wave
equation, see \cite{Ralston:1982, KKL:2001}. For the history of
Gaussian beams, see \cite{Ralston:1982, Popov:2002}.
In Section \ref{sec:Closure} we prove the main result Theorem
\ref{thm:mainResult}, and in Section \ref{sec:uni} we describe a
number of cases where Question \ref{mainProblem} has a negative
answer.

This paper relies on a number of computations done with computer algebra. 
Further information about these can be found on the author's homepage.

\section{Maxwell's equations}
\label{mainSec}
By a \emph{manifold} $M$ we mean a second countable topological Hausdorff
space that is locally homeomorphic to $\setR^n$ with $C^\infty$-smooth
transition maps. All objects are assumed to be smooth where defined.  
Let $TM$ and $T^\ast M$ be the tangent and cotangent bundles,
respectively, and for $k\ge 1$, let $\Lambda^k(M)$ be the set of
$p$-covectors, so that $\Lambda^1(N)=T^\ast N$.  Let $\Omega^k_l(M)$
be $k\choose l$-tensors that are antisymmetric in their $k$ upper
indices and $l$ lower indices. In particular, let $\Omega^k(M)$ be the
set of $k$-forms. Let also $\vfield{M}$ be the set of vector fields,
and let $C^\infty(M)$ be the set of functions. By $\Omega^k(M)\times
\setR$ we denote the set of $k$-forms that depend smoothly on a
parameter $t\in \setR$.
By $T(M,\setC)$, $T^\ast(M,\setC)$, $\Lambda^p(M,\setC)$,
$\Omega^k_l(M,\setC)$ and $\vfield{M,\setC}$ we denote the
complexification of the above spaces where component may also
take complex values. Smooth complex valued functions are denoted
by $C^\infty(M,\setC)$.
%
The Einstein summing convention is used throughout. When writing tensors 
in local coordinates we assume that the components satisfy the same symmetries as
the tensor. 

We will use differential forms to write Maxwell's equations.  On a
$3$-manifold $M$, \emph{Maxwell equations} then read \cite{BH1996,
 Obu:2003}
\begin{eqnarray}
\label{max1}
dE &=& - \pd{B}{t} , \\  
\label{max2}
dH &=& \pd{D}{t} + J, \\
\label{max3}
dD &=& \rho, \\
\label{max4}
dB &=& 0,
\end{eqnarray}
for field quantities $E,H\in \Omega^1 (M)\times \setR$, $D,B\in
\Omega^2 (M)\times \setR$ and sources $J\in \Omega^2(M)\times \setR$
and $\rho \in \Omega^3(M)\times \setR$. Let us emphasise that
equations \eqref{max1}--\eqref{max4} are completely
differential-topological and do not depend on any additional
structure. (To be precise, the exterior derivative does depend on the
smooth structure of $M$. However, for a manifold $M$ of dimension
$1,2,3$ one can show that all smooth structures for $M$ are
diffeomorphic. 
For higher dimensions the analogue result is not true.  Even for
$\setR^4$ there are uncountably many non-diffeomorphic smooth
structures \cite[p.~255]{Scorpan:2005}.)

\subsection{Maxwell's equations on a $4$-manifold}
\label{sec:MaxOn4}
Suppose $E,D,B,H$ are time dependent forms $E,H\in \Omega^1(M)\times
\setR$ and $D,B\in \Omega^2(M)\times \setR$ and $N$ is the
$4$-manifold $N=\setR\times M$. Then we can define forms $F,G \in
\Omega^2(N)$ and $j\in \Omega^3(N)$,
\begin{eqnarray}
\label{Fdef}
 F &=& B + E\wedge dt, \\
\label{Gdef}
 G &=& D - H\wedge dt,\\
 j &=& \rho-J\wedge dt.
\end{eqnarray}
Now fields $E,D,B,H$ solve Maxwell's equations equations
\eqref{max1}--\eqref{max4} if and only if
\begin{eqnarray}
\label{max4A}
dF &=& 0, \\
\label{max4B}
dG &=& j,
\end{eqnarray}
where $d$ is the exterior derivative on $N$. More generally, if $N$ is
a $4$-manifold and $F,G, j$ are forms $F,G\in \Omega^2(N)$ and $j\in
\Omega^3(N)$ we say that $F,G$ solve \emph{Maxwell's equations} (for
a source $j$) when equations \eqref{max4A}--\eqref{max4B} hold.
By an \emph{electromagnetic medium} on $N$
we mean a map 
\begin{eqnarray*}
   \kappa \colon \Omega^2(N) &\to& \Omega^2(N).
\end{eqnarray*} 
We then say that $2$-forms $F,G\in \Omega^2(N)$ \emph{solve Maxwell's
  equations in medium $\kappa$} if  $F$ and $G$ satisfy equations
\eqref{max4A}--\eqref{max4B} and
\begin{eqnarray}
\label{FGchi}
  G &=& \kappa(F).
\end{eqnarray}
Equation \eqref{FGchi} is known as the \emph{constitutive equation}.
If $\kappa$ is invertible, 
it follows that one can eliminate half of the free variables in
Maxwell's equations \eqref{max4A}--\eqref{max4B}.
We assume that $\kappa$ is linear and local so that we can represent
$\kappa$ by an antisymmetric $2\choose 2$-tensor $\kappa \in
\Omega^2_2(N)$. If in coordinates $\{x^i\}_{i=0}^3$ for $N$ we have
\begin{eqnarray}
\label{eq:kappaLocal}
  \kappa &=& \frac 1 2 \kappa^{ij}_{lm} dx^l\otimes dx^m\otimes \pd{}{x^i}\otimes \pd{}{x^j}
\end{eqnarray}
and $F = F_{ij} dx^i \otimes dx^j$ and $G = G_{ij} dx^i \otimes dx^j$, 
then constitutive equation \eqref{FGchi} reads
\begin{eqnarray}
 \label{FGeq_loc}
   G_{ij} &=& \frac 1 2 \kappa_{ij}^{rs} F_{rs}.
\end{eqnarray}

\subsection{Decomposition of electromagnetic medium}
\label{media:decomp}
Let $N$ be a $4$-manifold. Then
at each point on $N$, a general antisymmetric $2\choose
2$-tensor depends on $36$ parameters. Such tensors canonically decompose
into three linear subspaces. The motivation for this decomposition is
that different components in the decomposition enter in different
parts of electromagnetics.  See \cite[Section D.1.3]{Obu:2003}.  
The below formulation is taken from \cite{Dahl:2009}.

If $\kappa\in \Omega^2_2(N)$ we define the
\emph{trace} of $\kappa$ as the smooth function $N\to \setR$ given by
\begin{eqnarray*}
\operatorname{trace} \kappa &=& \frac 1 2 \kappa_{ij}^{ij}
\end{eqnarray*}
when $\kappa$ is locally given by equation \eqref{eq:kappaLocal}.
Writing 
$\operatorname{Id}$ as in equation \eqref{eq:kappaLocal} gives
$\operatorname{Id}^{ij}_{rs}= \delta^i_r\delta^j_s-\delta^i_s\delta^j_r$, so 
$\operatorname{trace}\operatorname{Id} = 6$ when $\dim N=4$. 

\begin{proposition}[Decomposition of a  $2\choose 2$-tensors]
\label{theorem:Decomp}
Let $N$ be a $4$-manifold, and let
\begin{eqnarray*}
Z &=& \{ \kappa \in \Omega^2_2(N) : u\wedge \kappa(v) = \kappa(u)\wedge v \,\,\mbox{for all}\,\, u,v\in \Omega^2(N),\\
& & \quad\quad\quad\quad\quad\quad \operatorname{trace} \kappa = 0\},\\
W &=& \{ \kappa \in \Omega^2_2(N) : u\wedge \kappa(v) = -\kappa(u)\wedge v \,\,\mbox{for all}\,\, u,v\in \Omega^2(N)\} \\ 
 &=& \{ \kappa \in \Omega^2_2(N) : 
 u\wedge \kappa(v) = -\kappa(u)\wedge v \,\,\mbox{for all}\,\, u,v\in \Omega^2(N), \\
& & \quad\quad\quad\quad\quad\quad \operatorname{trace} \kappa = 0 \}, \\
U &=& \{ f \operatorname{Id}\in \Omega^2_2(N) : f\in C^\infty(N) \}.
\end{eqnarray*}
Then
\begin{eqnarray}
\label{AdecompSet}
 \Omega^2_2(N) &=& Z\,\,\oplus\,\, W \,\,\oplus\,\, U,
\end{eqnarray}
and pointwise, $\dim Z = 20$,  $\dim W = 15$ and  $\dim U = 1$.
\end{proposition}

If we write a $\kappa\in \Omega^2_2(N)$ as
\begin{eqnarray*}
  \kappa &=& \kappaI \,\,+ \,\,\kappaII\,\,+\,\,\kappaIII
\end{eqnarray*}
with $\kappaI\in Z$, $\kappaII\in W$, $\kappaIII\in U$, then we say that 
$\kappaI$ is the \emph{principal part},
$\kappaII$ is the \emph{skewon part},
$\kappaIII$ is the \emph{axion part} of $\kappa$.


\subsection{The Hodge star operator} 
\label{sec:Hodge}
By a \emph{pseudo-Riemann metric} on a manifold $M$ we mean a symmetric real
$0\choose 2$-tensor $g$ that is non-degenerate. If $M$ is not
connected we also assume that $g$ has constant signature.  
If $g$ is positive definite, we say that $g$ is a \emph{Riemann metric}.
By $\sharp$
and $\flat$ we denote the isomorphisms $\sharp\colon T^\ast M\to TM$
and $\flat\colon TM\to T^\ast M$. By $\setR$-linearity we extend $g$,
$\sharp$ and $\flat$ to complex arguments. Moreover, we extend $g$
also to covectors by setting $g(\xi,\eta)=g(\xi^\sharp,\eta^\sharp)$
when $\xi,\eta\in \Lambda^1_p(N,\setC)$.

Suppose $g$ is a pseudo-Riemann metric on a
orientable manifold $M$ with $n=\dim M\ge 2$.  For $p\in\{0,\ldots,
n\}$, the \emph{Hodge star operator} $\ast$ is the map $\ast\colon
\Omega^p(M)\to \Omega^{n-p}(M)$ defined as
\cite[p. 413]{AbrahamMarsdenRatiu:1988}
\begin{eqnarray*}
\label{hodgedef}
\ast(dx^{i_1} \wedge \cdots \wedge dx^{i_p}) &=& \frac{\sqrt{|\det g|}}{(n-p)!} g^{i_1 l_1}\cdots  g^{i_p l_p} \varepsilon_{l_1 \cdots l_p\, l_{p+1} \cdots l_n} dx^{l_{p+1}}\wedge \cdots \wedge  dx^{l_{n}},
\end{eqnarray*}
where $x^i$ are local coordinates in an oriented atlas,
$g=g_{ij}dx^i\otimes dx^j$, $\det g= \det g_{ij} $, $g^{ij}$ is the
$ij$th entry of $(g_{ij})^{-1}$, and $\varepsilon_{l_1\cdots l_n}$ is
the \emph{Levi-Civita permutation symbol}. We treat
$\varepsilon_{l_1\cdots l_n}$ as a purely combinatorial object (and
not as a tensor density). We also define $\varepsilon^{l_1\cdots l_n}=
\varepsilon_{l_1\cdots l_n}$.

If $g$ is a pseudo-Riemann metric on an oriented
$4$-manifold $N$, then the Hodge star operator for $g$ induces a
$2\choose 2$-tensor $\kappa=\ast_g\in \Omega^2_2(N)$. If $\kappa$ is
written as in equation \eqref{eq:kappaLocal} for local coordinates
$x^i$ then
\begin{eqnarray}
\label{eq:hodgeKappaLocal}
  \kappa^{ij}_{rs} &=& \sqrt{\vert g\vert} g^{ia}g^{jb} \varepsilon_{abrs}.
\end{eqnarray}

\begin{proposition} 
\label{prop:HodgeHasOnlyPrincipalPart}
Suppose $g$ is a pseudo-Riemann metric on an orientable $4$-manifold $N$. 
Then $\ast_g$ defines a $2\choose 2$-tensor with only a principal part.
\end{proposition}

\begin{proof} Let $\kappa$ be the $2\choose 2$-tensor induced by $\ast_g$.
Then $u\wedge \kappa (v) = \kappa(u)\wedge v$ for
  all $u,v\in \Omega^2(N)$ \cite[p. 412]{AbrahamMarsdenRatiu:1988}. By Theorem \ref{theorem:Decomp}
  it therefore suffices to prove that $\operatorname{trace} \kappa =
  0$. Let us fix $p\in N$ and let  $x^i$ are local coordinates for $N$ such
that $g\vert_p$ is diagonal. If $\kappa$ is written as in equation
\eqref{eq:kappaLocal} then equation \eqref{eq:hodgeKappaLocal} implies that
$\operatorname{trace}\kappa=\frac 1 2 \kappa^{ij}_{ij}=0$  since $g^{ij}$ is diagonal and 
$\varepsilon_{ijkl}$ is non-zero only when $ijkl$ are distinct. 
\end{proof}

A pseudo-Riemann metric $g$ is a \emph{Lorentz metric} if $M$ is
$4$-dimensional and $g$ has signature $(+---)$ or $(-+++)$. 
For a Lorentz metric, we define the \emph{null cone} at $p$ as the set
$
   \{\xi\in \Lambda^1_p(M,\setR) : g(\xi,\xi)=0\}.
$
Usually, the null cone is defined as a subset in the tangent bundle. The
motivation for treating the null-cone in the cotangent bundle is given by
equation \eqref{eq:FresIsNullCone}.

The next theorem shows that the conformal class of a Lorentz metric $g$ can
be represented either using the $2\choose 2$-tensor $\ast_g$ or the
null cone of $g$.

\begin{theorem} 
\label{prop:Principle}
Suppose $g,h$ are  Lorentz metrics on an orientable $4$-manifold $N$. Then
the following are equivalent:
\begin{enumerate}
\item 
\label{eq:xx1}
There exists a non-vanishing function $\lambda\in C^\infty(N)$ such that $h=\lambda g$. 
\item 
\label{eq:xx2}
$\ast_g = \ast_h$, where $\ast_g$ and $\ast_h$ are the $2\choose 2$-tensors defined by $g$ and $h$, respectively.
\item
\label{eq:xx3} 
$g$ and $h$ have the same null cones.
\end{enumerate}
\end{theorem}

\begin{proof}
  Implications \ref{eq:xx1} $\Rightarrow$ \ref{eq:xx2} and
  \ref{eq:xx1} $\Rightarrow$ \ref{eq:xx3} are clear.  Implication
  \ref{eq:xx2} $\Rightarrow$ \ref{eq:xx1} is proven in
  \cite[Theorem 1]{Dray:1989}, and implication \ref{eq:xx3} $\Rightarrow$
  \ref{eq:xx1} is proven in \cite[Theorem 3]{Ehrlich:1991}. 
  See also \cite{MingSan:2008}.
\end{proof}

\subsection{Decomposition of $\kappa$ into four $3\times 3$ matrices}
\label{sec:ABCDtransRules} 
Suppose $(x^0, x^1, x^2, x^3)$ are local coordinates for
$N=\setR\times M$ such that $x^0$ is the coordinate for $\setR$ and
$(x^1, x^2, x^3)$ are coordinates for $M$. If forms $F,G$ are given by
equations \eqref{Fdef}--\eqref{Gdef}, then
$$
 F_{i0} = E_i, \quad 
F_{ij} = B_{ij}, \quad 
 G_{i0} = -H_i, \quad 
G_{ij} = D_{ij}
$$
for all $i,j=1,2,3$ and equation \eqref{FGeq_loc} then reads
\begin{eqnarray}
\label{eq:kappaLocal_I}
H_i &=& -\kappa_{i0}^{r0} E_r - \frac 1 2 \kappa_{i0}^{rs} B_{rs}, \\
\label{eq:kappaLocal_II}
D_{ij} &=& \kappa_{ij}^{r0} E_r +\frac 1 2 \kappa_{ij}^{rs} B_{rs},
\end{eqnarray}
where $i,j=1,2,3$ and $r,s$ are summed over $1,2,3$.

Next we show that in coordinates $(x^0, x^1,
x^2, x^3)$ the tensor $\kappa$ is represented by four $3\times
3$-matrices. To do this, let $\ast$ is the Hodge star operator
induced by the Euclidean metric on $x^1, x^2, x^3$ so that $\ast dx^i
= \frac 1 2 \sum_{a,b=1}^3 \varepsilon^{iab} dx^a \wedge dx^b$. 
Thus $B=
\sum_{i=1}^3 B^i \ast\!dx^i$ 
where $B^i =\frac 1 2 \varepsilon^{ijk} B_{jk}$ and $B_{mn} =
\varepsilon_{imn} B^i$.  In the same way we define $D^1, D^2, D^3$.  Now
components $D^i$ and $B^i$ represent $2$-forms $D$ and $B$ in the
basis $\{\ast dx^i \}_{i=1}^3$, and by equations
\eqref{eq:kappaLocal_I}--\eqref{eq:kappaLocal_II},
\begin{eqnarray}
\label{eq:kappaLocal_Ix}
  H_i &=& \cC^r\, _i (-E_r) + \cB_{ri}  B^{r}, \\
\label{eq:kappaLocal_IIx}
  D^{i} &=& \cA^{ri} (-E_r) +\cD_r\,^i B^{r},
\end{eqnarray}
where $i\in \{1,2,3\}$, $r$ is summed over $1,2,3$, and
$$
   \cC^r\,_i     = \kappa^{r0}_{i0}, \quad 
   \cB_{ri}        =  - \frac 1 2 \varepsilon_{rab} \kappa^{ab}_{i0}, \quad 
  \cA^{ri}       = -\frac 1 2 \varepsilon^{iab} \kappa^{r0}_{ab}, \quad  
  \cD_r\,\!^i     = \frac 1 4 \varepsilon_{rmn} \varepsilon^{iab} \kappa^{mn}_{ab}. 
$$
Here $r$ index rows and $i$ index columns in $3\times 3$ matrices 
$\cA, \cB, \cC, \cD$. 
Inverting the relations gives
$$
    \kappa^{0r}_{0i}    = \cC^r\,\! _i, \quad
    \kappa^{ij}_{0r}     = \varepsilon^{kij} \cB_{kr}, \quad
    \kappa^{0i}_{rs}    = \varepsilon_{krs} \cA^{ik}, \quad
    \kappa^{ij}_{rs}     = \varepsilon_{krs} \varepsilon^{lij} \cD_l\,\!^k,
$$
where $i,j,r,s\in \{1,2,3\}$ and $k, l$ are summed over $1,2,3$.

The above matrices $\cA, \cB, \cC, \cD$ coincide with the matrices
$\cA, \cB, \cC, \cD$ defined in \cite[Section D.1.6]{Obu:2003} and
\cite{Rubilar2002}.  Since these matrices are only part of tensor
$\kappa$, they do not transform in a simple way under a general
coordinate transformation in $N$ (see equations D.5.28--D.5.30 in
\cite{Obu:2003}).  However, if $\{x^i\}_{i=0}^3$ and $\{\widetilde
x^i\}_{i=0}^3$ are overlapping coordinates such that
\begin{eqnarray*}
  \widetilde x^0 &=& x^0,\\
  \widetilde x^i &=& \widetilde x^i(x^1, x^2, x^3), \quad i\in \{1,2,3\}.
\end{eqnarray*}
Then we have transformation rules 
\begin{eqnarray}
   \widetilde \cC^r\,_i  &=&   
\label{eq:transRuleC}
\cC^a\,_b \pd{ x^b}{\widetilde x^i}\pd{\widetilde x^r}{ x^a}, \\
   \widetilde \cB_{ri}    &=& 
\label{eq:transRuleB}
\det\left( \pd{\widetilde x^m}{ x^n}\right)  \cB_{ab} \pd{x^a}{\widetilde  x^r}\pd{x^b}{\widetilde  x^i}, \\
   {\widetilde \cA}^{ri}    &=& 
\label{eq:transRuleA}
\det\left( \pd{ x^m}{ \widetilde x^n}\right)  \cA^{ab} \pd{\widetilde x^r}{  x^a}\pd{\widetilde x^i}{  x^b}, \\
 \widetilde \cD_r\,\!^i     &=&  
\label{eq:transRuleD}
\cD_a\,\!^b   \pd{\widetilde  x^i}{ x^b}\pd{ x^a}{\widetilde  x^r}.
\end{eqnarray}

If $\kappaII=0$ then Proposition \ref{theorem:Decomp} implies that
$\kappa$ is pointwise determined by $21$ coefficients. The next
proposition shows that these coefficients can pointwise be reduced to
$18$ when the coordinates are chosen suitably.

\begin{proposition}
\label{prop:localReduction}
Suppose $N$ is a $4$-manifold and $\kappa\in \Omega^2_2(N)$.
Then 
\begin{enumerate}
\item 
\label{cc:AA}
$\kappa$ has no skewon component if and only if locally
$$
   \cA = \cA^T, \quad 
   \cB = \cB^T, \quad 
    \cC=\cD^T,
$$
where $^T$ is the matrix transpose, and $\cA, \cB, \cC, \cD$ are defined as above.
\item
\label{cc:BB}
Let $p\in N$. If $\kappa$ has no skewon component, then there are
local coordinates around $p$ such that $\cA$ is diagonal at $p$.
\item 
\label{cc:CC}
Let $p\in N$. If $\kappa$ has no skewon component and $g$ is a Lorentz
metric on $N$ there are local coordinates around $p$ such that
$\cA\vert_p$ is diagonal and for some $k\in \{\pm 1\}$ we have
$g\vert_p = k\operatorname{diag}(-1,1,1,1)$.
\end{enumerate}
\end{proposition}

\begin{proof} 
Part \ref{cc:AA} follows by \cite[Equation
D.1.100]{Obu:2003}. 
Since we can always introduce a Lorentz metric in local coordinates for $N$, part 
 \ref{cc:BB} will follow from part \ref{cc:CC}.
 For part \ref{cc:CC}, let $x^i$ be coordinates around $p$ such that
 $g\vert_p = k\operatorname{diag}(-1,1,1,1)$ for $k\in \{\pm 1\}$. By
 \ref{cc:AA}, matrix $\cA\vert_p$ is symmetric, so we can find an
 orthogonal $3\times 3$ matrix $P=(P^i\,\,_j)_{ij}$ such that $P \cA
 P^T$ is diagonal and $\det P=1$. A suitable coordinate system is
 given by $\widetilde x^0 = x^0$ and $\widetilde x^i =\sum_{j=1}^3
 P^i\,_j x^j.$
\proofread{\textbf{Theorem:} If $A$ is a symmetric square matrix, then there exists a 
orthogonal matrix $Q$ such that $A = Q^T D Q = Q^{-1} D Q$  where $D$ is the diagonal matrix with
the eigenvalues of $A$ on the diagonal.
}
\end{proof}

\section{Geometric optics solutions}
\label{sec:GOS}
Let $\kappa\in \Omega^2_2(N)$ on a $4$-manifold $N$, and let $F$ and $G$ be asymptotic sums
\begin{eqnarray}
\label{eq:FGtrial}
   F = \REAL\left\{e^{iP \Phi} \sum_{k=0}^\infty \frac{A_k}{(iP)^k} \right\},\quad
   G = \REAL\left\{e^{iP \Phi} \sum_{k=0}^\infty \frac{B_k}{(iP)^k} \right\},
\end{eqnarray}
where $P>0$ is a constant, $\Phi \in C^\infty(N,\setC)$ and $A_k, B_k \in \Omega^2(N,\setC)$. 
Substituting $F$ and $G$ into the sourceless Maxwell equations and differentiating termwise shows that
$F$ and $G$ form an asymptotic solution provided that
\begin{eqnarray}
\label{eq:As1}
d\Phi \wedge  A_0 &=&0,\\
\label{eq:As2}
d\Phi \wedge  B_0 &=&0,\\
\label{eq:kConst}
B_k &=& \kappa A_k, \\
\label{eq:Trans1}
d\Phi \wedge  A_{k+1} + dA_{k} &=&0,\\
\label{eq:Trans2}
d\Phi \wedge  B_{k+1} + dB_{k} &=&0,
\quad k=0,1,\ldots.
\end{eqnarray}
In equation \eqref{eq:kConst} we treat $\kappa$ as a linear map $\kappa\colon \Omega^2(N,\setC)\to\Omega^2(N,\setC)$.
In equation \eqref{eq:FGtrial} function $\Phi$ is called a \emph{phase function}, and
forms $A_k, B_k$ are called \emph{amplitudes}.
We will assume that $\IMAG \Phi\ge 0$, so
that $F$ and $G$ remain
bounded even if we take $P\to \infty$.  

\begin{lemma}  
\label{lemma:solv}
Suppose $N$ is a smooth manifold, and let $q$ be a $1$-form $q\in
\Omega^1(N,\setC)$ that is nowhere zero.
\begin{enumerate}
\item 
\label{lemmaI}
If $q\wedge A = 0$ for some $A\in \Omega^k(N,\setC)$ where $k\ge 1$,
then there exists a $(k-1)$-form $a\in \Omega^{k-1}(N, \setC)$ such that $A =
q\wedge a$.
\item 
\label{lemmaII}
If $q\wedge a = q\wedge a'$ for some $a,a' \in \Omega^1(N, \setC)$, then $a = a' + f q$
for some $f\in C^\infty(N,\setC)$.
\end{enumerate}
\end{lemma}

\begin{proof}
  Let $\sharp$ be the isomorphism $T^\ast N\to TN$ induced by an
  auxiliary (positive definite) Riemann metric on $N$, and let $\Vert
  \cdot \Vert$ be the induced norms on $TN$ and $T^\ast N$. Let also $q = \alpha +
  i\beta$, where $\alpha = \REAL q$ and $\beta = \IMAG q$. Then vector
  field $X\in \vfield{N, \setC}$ given by
\begin{eqnarray*}
   X &=&\frac{ \alpha^\sharp -  i \beta^\sharp}{ \Vert \alpha\Vert^2 + \Vert \beta\Vert^2}
\end{eqnarray*}
satisfies $q(X)=1$. Contracting $q\wedge A=0$ by $X$ 
gives part \ref{lemmaI}.
Part \ref{lemmaII} follows by taking $A=a-a'$ in part \ref{lemmaI}.
\end{proof}  

In this work we will only analyse the leading amplitudes $A_0$ and
$B_0$.  However, since $B_0 = \kappa(A_0)$, it suffices to study $A_0$
in more detail.  Let us assume that $\Phi$, $A_0$ and $B_0$ solve
equations \eqref{eq:As1}--\eqref{eq:kConst}. Then Lemma
\ref{lemma:solv} \ref{lemmaI} implies that there exists a $1$-form
$a_0 \in \Omega^1(N,\setC)$ such that $ A_0 = d\Phi\wedge a_0, $
whence
\begin{eqnarray} 
\label{eq:NeqA0xx}
d\Phi \wedge \kappa( d\Phi \wedge a_0) &=& 0. 
\end{eqnarray}

For $N=\setR\times M$ where $M$ is a $3$-manifold and for special
choices for $\kappa$, $\Phi$ and amplitudes $A_k, B_k$, equation
\eqref{eq:FGtrial} define an electromagnetic \emph{Gaussian beam} (see
Section \ref {sec:Intro}). 
In this setting, $\Phi\vert_p$ and $d\Phi\vert_p$ are both
real when $p$ is at a centre of a Gaussian beam.
With the above as motivation we will hereafter only study equation
\eqref{eq:NeqA0xx} at a point $p\in N$ where $d\Phi$ is real.  From
equation \eqref{eq:FGtrial} we then see that $d\Phi\vert_p$
is the direction of most rapid oscillation
(or direction of propagation) for $F$.  Since $A_0 = d\Phi\wedge a_0$, the $1$-form $a_0$, in turn,
determines the polarisation of the solution in equation
\eqref{eq:FGtrial}.  Equation \eqref{eq:NeqA0xx} is thus a condition
that constrains possible polarisations once the direction of
propagation is known.  Since equation \eqref{eq:NeqA0xx} is a linear
in $a_0$, we may study the dimension of the the solution space for
$a_0$.
To do this, let  $\xi\in \Lambda^1_p(N)$ for some $p\in N$ and 
for $\xi$ let $L_\xi$ be the linear map
 $L_\xi \colon \Lambda^1_p(N)\to \Lambda^3_p(N)$,
\begin{eqnarray}
\label{Ldef}
  L_\xi(\alpha) &=& \xi \wedge \kappa( \xi \wedge \alpha), \quad \alpha \in \Lambda_p^{1}(N).
\end{eqnarray}
We have $\xi \in \operatorname{ker} L_\xi$. For all $\xi \in \Lambda^1_p N\slaz$ we can then
find a (non-unique) vector subspace $V_{\xi}\subset \Lambda^1_pN$ such that
\begin{eqnarray}
\label{eq:directSumKer}
  \kerOp L_{\xi} &=& V_{\xi}  \, \oplus \, \operatorname{span} \xi.
\end{eqnarray}

Let $\xi = d\Phi\vert_p$ be nonzero. Then $V_{\xi}\slaz$ parameterises
possible $a_0$ that solve equation \eqref{eq:NeqA0xx} and for which
$A_0 = d\Phi\wedge a_0$ is nonzero.  For a general $\kappa\in
\Omega^2_2(N)$ and $\xi\in \Lambda^1 (N)\slaz$ we can have $\dim
V_{\xi}\in\{0,1,2, 3\}$: Proposition \ref{prop:FresnelForRiemann} will
show that $\dim V_\xi$ can be $0$ or $2$, Example \ref{ex:dimVxi=1}
shows that $\dim V_\xi$ can be $1$, and the next proposition
characterise $\kappa\vert_p$ when $\dim V_\xi=3$ for all $\xi\in
\Lambda_p^1 (N)\slaz$.

\begin{proposition}
Let $\kappa\in \Omega^2_2(N)$ on a $4$-manifold $N$ and let $p\in N$.  
Then the following are equivalent:
\begin{enumerate}
\item $\kappa\vert_p$ is of axion type. 
\label{le:I}
\item $\dim V_\xi=3$ for all $\xi\in \Lambda_p^1(N)\slaz$.
\label{le:II}
\end{enumerate}
\end{proposition}

\begin{proof}
  Implication \ref{le:I} $\Rightarrow$ \ref{le:II} is clear. For the
  converse direction suppose that \ref{le:II} holds and $x^i$ are
  local coordinates around $p$. It follows that
\begin{eqnarray*}
     \zeta\wedge \xi\wedge \kappa(\xi\wedge \alpha)&=&0, \quad \alpha,\xi,\zeta\in \Lambda^1_p(N).
\end{eqnarray*}
If locally $\xi=\xi_i dx^i\vert_p$ then
$
    \xi_i \xi_j \kappa^{ir}_{ab} \varepsilon^{jsab}=0.
$
Differentiating with respect to ${\xi_c}$ and ${\xi_d}$ gives
\begin{eqnarray*}
    \kappa^{cr}_{ab} \varepsilon^{dsab} + \kappa^{dr}_{ab} \varepsilon^{csab} &=&0.
\end{eqnarray*}
With computer algebra it follows that
$\kappa =\frac 1 6 \operatorname{trace} \kappa\,\, \operatorname{Id}$
and \ref{le:I} follows.
\end{proof} 

\subsection{Fresnel surface}
Let 
$\kappa\in \Omega^2_2(N)$ on a $4$-manifold $N$. 
If $\kappa$ is locally given by equation \eqref{eq:kappaLocal} in coordinates $x^i$, let 
\begin{eqnarray*}
  \cG^{ijkl}_0 &=& \frac 1 {48} 
\kappa^{a_1 a_2}_{b_1 b_2} 
\kappa^{a_3 i}_{b_3 b_4} 
\kappa^{a_4 j}_{b_5 b_6} 
\varepsilon^{b_1 b_2 b_5 k} 
\varepsilon^{b_3 b_4 b_6 l} 
\varepsilon_{a_1 a_2 a_3 a_4}.
\end{eqnarray*}
In overlapping coordinates $\{\widetilde x^i\}$, these coefficients
transform as
\begin{eqnarray}
\label{eq:TRtrans}
  \widetilde \cG_0^{ijkl} &=& \det \left(\pd{x^r}{\widetilde x^s}\right)\, \cG_0^{abcd} \pd{\widetilde x^i}{x^a} \pd{\widetilde x^j}{x^b}\pd{\widetilde x^k}{x^c}\pd{\widetilde x^l}{x^d}.
\end{eqnarray}
Thus components $\cG^{ijkl}_0$ define a tensor density $\cG_0$ on $N$ of
weight $1$. The \emph{Tamm-Rubilar tensor density} \cite{Rubilar2002,
  Obu:2003} is the symmetric part of $\cG_0$ and we denote this tensor density by $\cG$. 
In coordinates,
$\cG^{ijkl} = \cG^{(ijkl)}_0$, where parenthesis indicate that 
indices $ijkl$ are symmetrised with scaling $1/4!$.  Using tensor
density $\cG$, the \emph{Fresnel surface} at a point $p\in N$
is defined as 
\begin{eqnarray}
\label{eq:Fr}
 F_p &=& \{\xi\in \Lambda^1_p(N) : \cG^{ijkl} \xi_i \xi_j \xi_k \xi_l  = 0\}.
\end{eqnarray}
By equation \eqref{eq:TRtrans}, the definition of $F_p$ does not
depend on local coordinates.  Let $F$ be the disjoint union of all
Fresnel surfaces, $F=\coprod_{p\in N} F_p$. To indicate that $F_p$ and
$F$ depend on $\kappa$ we also write $F_p(\kappa)$ and $F(\kappa)$.

If $\xi\in F_p$ then $\lambda \xi\in F_p$ for all $\lambda\in \setR$.  In
particular $0\in F_p$ for each $p\in N$. 
When $\cG\vert_p$ is non-zero, equation \eqref{eq:Fr} shows that $F_p$
is a fourth order surface in $\Lambda^1_p(N)$, so $F_p$ may contain
non-smooth self intersections.

\begin{theorem}
\label{eq:ThmDimTR}
  Suppose $N$ is a $4$-manifold and $\kappa\in \Omega^2_2(N)$.  If
  $\xi\in \Lambda^1_p(N)$ is non-zero, then the following are
  equivalent:
\begin{enumerate}
\item 
\label{thm:TR_i}
$\dim V_\xi\ge 1$ where $V_\xi$ are defined as in equation \eqref{eq:directSumKer}.
\item 
\label{thm:TR_ii}
$\xi$ belongs to the Fresnel surface  $F_p\subset \Lambda^1_p(N)$.
\end{enumerate}
\end{theorem}

\begin{proof}
  Let $\{x^i\}_{i=0}^3$ be coordinates around $p$ such that $dx^0\vert_p=\xi$ and
  let $\ast$ be the Hodge star operator induced by the Euclidean Riemann metric $g_{ij} =\delta_{ij}$
  in these coordinates. Let $P\colon \Lambda_p^1(N)\to
  \Lambda_p^1(N)$ be the map $P=2\ast\circ L_\xi$. Then locally
\begin{eqnarray*}
  P(\alpha) &=& \sum_{j=0}^3 \alpha_i \varepsilon^{0abj} \kappa^{0i}_{ab} dx^j,
\end{eqnarray*}
where $\alpha = \alpha_i dx^i\vert_p$ and $\kappa^{ij}_{ab}$ are defined as in equation 
\eqref{eq:kappaLocal}. It follows that in the
basis $\{dx^i\vert_p\}_{i=0}^3$, the map $P$ is represented by the $4\times 4$ matrix $\operatorname{diag}(0, Q)$,
where $Q$ is the $3\times 3$ matrix  $Q^{ij} = \varepsilon^{0abj} \kappa^{0i}_{ab}$, $i,j\in \{1,2,3\}$.
Now $\dim V_\xi\ge 1$ is equivalent with $\dim \operatorname{ker} P\ge 2$ which 
is equivalent with $\det Q = 0$. 
Writing out $\det Q=0$ using
\begin{eqnarray*}
   \det Q &=& \frac 1 {3!} \varepsilon_{abc} \varepsilon_{ijk} Q^{ai} Q^{bj} Q^{ck}
\end{eqnarray*}
gives $\cG^{ijkl}\xi_i\xi_j\xi_k\xi_l=0$. We omit the proof of the
last step which can be found in \cite{Rubilar2002} and
\cite[p. 267 -- 268]{Obu:2003}.
\proofread{ Let us check this: First define
  $\chi^{abcd} = \frac 1 2 \varepsilon^{abpq} \kappa_{pq}^{cd}$. Then
  $Q^{ij} = 2 \chi^{0j0i}$ and $\chi^{abcd}$ is antisymmetric in $ab$
  and $cd$. By anti-symmetry,
$$
   \varepsilon_{0bcd}  \chi^{0b0r} = 
\frac 1 2 
  \varepsilon_{\alpha \beta cd}  \chi^{\alpha \beta 0r}.
$$
Since $\varepsilon_{abc} = \varepsilon_{0abc}$ it follows that 
\begin{eqnarray*}
    \det Q 
              &=& \varepsilon_{0bcd} \varepsilon_{0rst} \chi^{0b0r} 
\chi^{0c0s} \chi^{0d0t} \\
              &=& \frac 1 2 
\varepsilon_{\alpha \beta cd} \varepsilon_{r s0t} \chi^{\alpha \beta0r} \chi^{0c0s} \chi^{0d0t}\\
              &=& \frac 1 4
\varepsilon_{\alpha \beta cd} \varepsilon_{r s\mu \tau } \chi^{\alpha \beta0r} \chi^{0c0s} \chi^{0d\mu \tau }\\ 
              &=& \frac 1 4
\varepsilon_{\alpha \beta \gamma\delta} \varepsilon_{\rho \sigma\mu \tau } \chi^{\alpha \beta0\rho} \chi^{0\gamma 0\sigma} \chi^{0\delta \mu \tau }\\ 
\end{eqnarray*}
}
\end{proof}

Suppose $g$ is a pseudo-Riemann metric on an orientable $4$-manifold $N$ and 
$\kappa \in \Omega^2_2(N)$. Then $g$ and $\kappa$ define a symmetric 
$4\choose 0$-tensor on $N$ by
\begin{eqnarray}
\label{eq:Tg}
  \cG_{g,\kappa} &=& \frac{1}{\sqrt{\vert \operatorname{det} g\vert}}\,\, \cG^{ijkl} \pd{}{x^i}\otimes \pd{}{x^j}\otimes \pd{}{x^k}\otimes \pd{}{x^l},
\end{eqnarray}
where $\cG^{ijkl}$ are local components of the Tamm-Rubilar tensor
density for $\kappa$, and $x^i$ are coordinates in an oriented
atlas for $N$.

A key property of symmetric $p\choose 0$-tensors is that they are
completely determined by their values on the diagonal
\cite{Mujica:2006, PunziEtAl:2009}.  For symmetric $4\choose
0$-tensors on a $4$-manifold (like $\cG_{g,\kappa}$), the precise
statement is contained in the following polarisation
identity.
 
\begin{proposition}
\label{prop:polarId}
Suppose $L$ is a symmetric $4\choose 0$-tensor on a $4$-manifold $N$. If 
$x_1, x_2, x_3, x_4\in \Lambda^1_p(N)$ then 
\begin{eqnarray*}
L(x_1, x_2, x_3, x_4) &=& \frac {1}{4! 2^4} \sum_{\theta_i\in \{\pm1\}} 
\theta_1 \theta_2 \theta_3 \theta_4 
\, L(\sum_{i=1}^4\theta_i x_i, \ldots,  \sum_{i=1}^4\theta_i x_i).
\end{eqnarray*}
\end{proposition}

\subsection{Electromagnetic medium induced by a Hodge star operator}
In Proposition \ref{prop:HodgeHasOnlyPrincipalPart} we saw that a
pseudo-Riemann metric on a $4$-manifold induces a $2\choose 2$-tensor
$\kappa$ with only a principal part.  The next example shows how
standard isotropic electromagnetic medium can be modelled using a
Lorentz metric on $\setR^4$.
 
\begin{example} 
\label{eq:ABCDminkowski}
On $N=\setR\times \setR^3$ let $\kappa$ be the $2\choose 2$-tensor determined 
$3\times 3$ matrices
$$
\cA = -\epsilon \operatorname{Id}, \quad
\cB = \mu^{-1}\operatorname{Id}, \quad
\cC =\cD = 0,
$$
where $\epsilon,\mu\colon \setR^3\to (0,\infty)$.
Then constitutive equations \eqref{eq:kappaLocal_Ix}--\eqref{eq:kappaLocal_IIx} are 
equivalent with the isotropic constitutive equations
\begin{eqnarray}
\label{eq:isotropicEq1}
D &=& \epsilon \ast_0 E, \\
\label{eq:isotropicEq2}
B &=&  \mu \ast_0 H,
\end{eqnarray}
where $\epsilon$ is the \emph{permittivity} and $\mu$ is the
\emph{permeability} of the medium and $\ast_0$ is the Hodge star operator
induced by the Euclidean metric on $\setR^3$.
If $\kappa$ is the $2\choose 2$-tensor defined as 
$\kappa = \sqrt{\frac{\epsilon}{\mu}} \ast_g$ where $g$ is the Lorentz metric 
$
  g =  \operatorname{diag}(-\frac{1}{\epsilon \mu},1,1,1)
$, then equations \eqref{eq:isotropicEq1}--\eqref{eq:isotropicEq2} are equivalent
with equation \eqref{FGchi}.
\proofBox
\end{example}

The next proposition shows that if $g$ is a pseudo-Riemann metric with
signature $(++++)$ or $(----)$ then the medium with $\kappa =
\ast_g$ has no asymptotic solutions. That is, if $d\Phi\vert_p$ is
non-zero, then equation \eqref{eq:NeqA0xx} implies that
$A_0\vert_p=0$. The proposition also shows that if $\kappa = \ast_g$
for an indefinite metric $g$, then $A_0$ can be non-zero only when
$d\Phi\vert_p$ is a \emph{null covector}, that is, when
$g(d\Phi\vert_p,d\Phi\vert_p)=0$.

Let $\operatorname{sgn}\colon \setR\to \{-1,+1\}$ be the \emph{sign function},
$\operatorname{sgn} x= -1$ for $x<0$, 
$\operatorname{sgn} x= 0$ for $x=0$ and 
$\operatorname{sgn} x= 1$ for $x>0$.

\begin{proposition} 
\label{eq:ThgExp}
Let $g$ and $h$ be pseudo-Riemann metrics on $N$ on an orientable
$4$-manifold $N$. Then
\begin{eqnarray*}
  \cG_{h,\ast g}(\xi,\xi,\xi,\xi) &=& \operatorname{sgn}(\det g)\,  \sqrt{\frac{\vert \det g\vert }{\vert \det h\vert} }\left( g(\xi,\xi)\right)^2, \quad \xi\in \Lambda^1( N).
\end{eqnarray*}
Thus the Fresnel surface induced by the $2\choose 2$-tensor $\ast_g$ is given by
\begin{eqnarray*}
\label{eq:FresIsNullCone}
  F(\ast_g) &=& \{ \xi\in \Lambda^1(N) : g(\xi,\xi)=0\}.
\end{eqnarray*}
\end{proposition}

\begin{proof} 
Let $\cG^{ijkl}$ be components for the Tamm-Rubilar tensor density for $\ast_g$.
Computer algebra then gives
\begin{eqnarray*}
  \cG^{abcd} \xi_a \xi_b \xi_c \xi_d &=& \operatorname{sgn}(\det g)\, 
\sqrt{\vert \det g\vert} \left( g(\xi,\xi)\right)^2,
\end{eqnarray*}
where $\xi=\xi_a dx^a$ and the claim follows by equation \eqref{eq:Tg}.
\end{proof}

We know that a general plane wave in homogeneous isotropic medium in
$\setR^3$ can be written as a sum of two circularly polarised plane
waves with opposite handedness. 
The \emph{Bohren decomposition} generalise this classical result to
electromagnetic fields in homogeneous isotropic chiral medium
\cite{LiSiTrVi:1994}.  The \emph{Moses decomposition}, or
\emph{helicity decomposition}, further generalise this decomposition
to arbitrary vector fields on $\setR^3$. For a decomposition of
Maxwell's equations using this last decomposition, see
\cite{Moses:1971, Dahl2004}.  In all of these cases, 
an electromagnetic wave can be polarised in two
different ways. Part \ref{prop:FrLoI} in the next proposition shows
that this is also the case for asymptotic solutions as defined above
when the medium is given by the Hodge star operator of a indefinite metric.

\begin{proposition} 
\label{prop:FresnelForRiemann}
Let $N$ be an orientable $4$-dimensional manifold, and let $\kappa \in \Omega^2_2(N)$
the $2\choose 2$-tensor $\kappa =\ast_g$ induced by a pseudo-Riemann metric $g$ on $N$. 
\begin{enumerate}
\item 
\label{prop:FrLoI}
If $\xi \in \Lambda^1(N)$ is non-zero, and $V_\xi$ is as in equation \eqref{eq:directSumKer}, then
\begin{eqnarray*}
\dim V_{\xi} &=& \begin{cases} 2, &\mbox{when} \, \, \xi\in F(\kappa), \\ 
   0, &\mbox{when} \, \, \xi\notin F(\kappa).
\end{cases}
\end{eqnarray*}
\item 
\label{prop:FrLoII}
If $\xi \in F(\kappa)$ is non-zero, and $L_\xi$ is as in equation \eqref{Ldef} then 
\begin{eqnarray*}
   \operatorname{ker}L_\xi &=& \xi^\perp,
\end{eqnarray*}
where $\xi^\perp = \{ \alpha\in \Lambda^1(N) : g(\alpha, \xi)=0\}$. Thus, 
for any choice of $V_\xi$ in equation \eqref{eq:directSumKer} we have $V_\xi\subset \xi^\perp$.
\end{enumerate}
\end{proposition}

\begin{proof} 
Let $p$ be the basepoint of $\xi$ and let
  $\{x^i\}_{i=0}^3$ are local coordinates for $N$ around $p$ such that
  $g=g_{ij} dx^i\otimes dx^j$ and $g_{ij}\vert_p$ is diagonal with
  entries $\pm 1$.  We know that $\kappa^2 = \ast_g^2 = (-1)^\sigma
  \operatorname{Id}$, where $\sigma$ is the
  \emph{index} of $g$ \cite[p. 412] {AbrahamMarsdenRatiu:1988}.  If
  $\alpha \in \Lambda^1_p(N)$, equations \eqref{Ldef} and
  \eqref{eq:hodgeKappaLocal} imply that
\begin{eqnarray}
\nonumber
L_\xi(\alpha)&=& \frac 1 2 \xi_r \xi_s \alpha_i g^{r a} g^{i b} \varepsilon_{abcd} dx^s \wedge dx^{c}\wedge dx^{d} \\
\nonumber
&=& \det g \, (-1)^{\sigma} \alpha_i H^{ir} g_{rs} \ast dx^s, 
\label{eq:transformToLocal}
\end{eqnarray}
where $\xi = \xi_i dx^i\vert_p$ and $\alpha = \alpha_i dx^i\vert_p$ and 
\begin{eqnarray}
H^{ir} &=& g(\xi,\xi) g^{ir} -  \xi_a g^{ai}\xi_b g^{br}.
\label{eq:Cik}
\end{eqnarray}
For part \ref{prop:FrLoI},  equations \eqref{eq:transformToLocal} and \eqref{eq:directSumKer}
imply that $\dim V_\xi = \dim \kerOp H-1$ where $H$ is the $4\times 4$
matrix with entries $H^{ij}$.  Let $\sigma(H)$ denote the spectrum of
$H$ with eigenvalues repeated according to their algebraic
multiplicity.  With computer algebra we find that
\begin{eqnarray*}
\sigma(H) &=& \left(0, C_1 g(\xi,\xi), C_2 g(\xi,\xi), C_3 \sum_{i=0}^3 \xi_i^2\right),
\end{eqnarray*}
where $C_i\in \{\pm 1\}$ are constants that depend only the signature
of $g$.  Now part \ref{prop:FrLoI} follows by Proposition
\ref{eq:ThgExp} and since algebraic and geometric multiplicity of an
eigenvalue coincide for symmetric matrices \cite[p. 260]{Szabo:2002}.
For part \ref{prop:FrLoII}, equality $\operatorname{ker}
L_\xi=\xi^\perp$ follows from the local representation of $L_\xi$ in
equation \eqref{eq:Cik}.
\end{proof}

The next example shows that the case
$\operatorname{dim} V_\xi = 1$ is possible in equation \eqref{eq:directSumKer}.
The medium defined by equations \eqref{eq:biaxial}
is called 
a \emph{biaxial crystal} \cite[Section 15.3.3]{BornWolf}.

\begin{example}
\label{ex:dimVxi=1} 
On $M=\setR\times \setR^3$, let $\kappa\in \Omega^2_2(M)$ be defined by
\begin{eqnarray}
\label{eq:biaxial}
  \cA = -\operatorname{diag}(1,2,3), \quad
  \cB = \operatorname{Id}, \quad
  \cC=\cD=0.
\end{eqnarray} 
Then the Fresnel equation reads
\begin{eqnarray}
\label{eq:FresnelEqXXY}
6 \xi_0^4 
- \xi_0^2 (5 \xi_1^2 + 8 \xi_2^2 + 9 \xi_3^2)
+ (\xi_1^2 + \xi_2^2 + \xi_3^2) (\xi_1^2 + 2 \xi_2^2 + 3 \xi_3^2)  &=& 0.
\end{eqnarray}
Let $S$ be the solution set in $\setR^3$ to the above equation when
$\xi_0=1$. By equation \eqref{eq:FresnelEqXXY}, it is clear that $S$
is mirror symmetric about the $\xi_1\xi_2$, $\xi_1\xi_3$ and
$\xi_2\xi_3$ coordinate planes. Figure \ref{fig:singular} below
illustrates $S$ in the quadrant $\xi_1\ge 0, \xi_2\ge 0, \xi_3\ge 0$,
and in this quadrant we see that $S$ has one singular point
$\xi_{\operatorname{sing}}\in S$. 
\begin{center}
\begin{figure}[!ht]
\includegraphics[width= 0.65\textwidth]{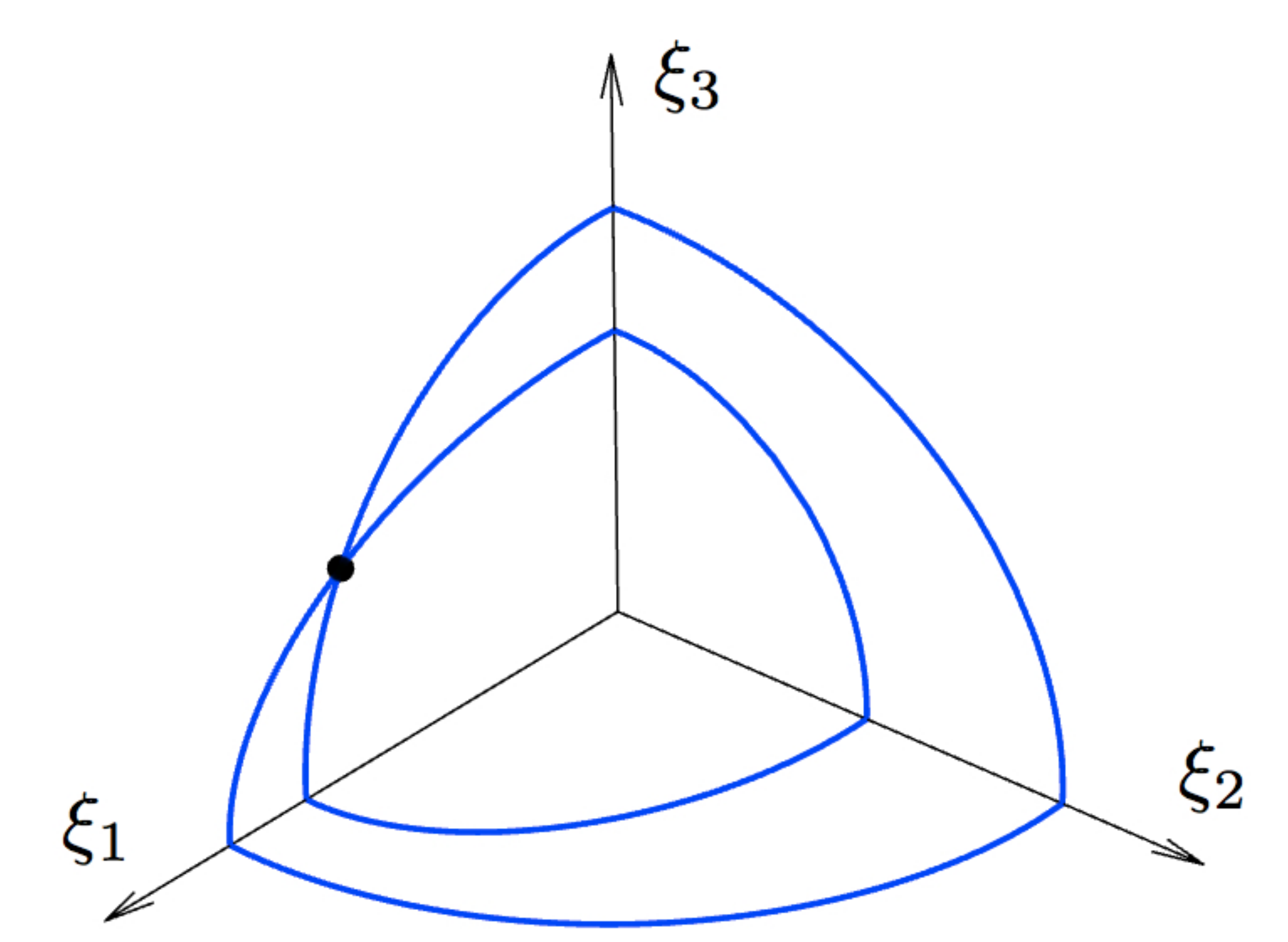}
\caption{One quadrant in $\setR^3$ of a Fresnel surface with a singular point  illustrated by a dot.}
\label{fig:singular}
\end{figure} 
\end{center}
Surface $S$ is defined implicitly by $f(\xi_1, \xi_2,\xi_3)=0$ and
singular points are characterised by $\nabla f=0$. This yields
$\xi_{\operatorname{sing}} = (
\sqrt{\frac{3}{2}},0,\frac{1}{\sqrt{2}})$. (For an alternative way to
solve this point, see \cite[Lemma 4.2 \emph{(iii)}]{Dahl2004}.)  Using
computer algebra and the arguments used to prove Theorem
\ref{eq:ThmDimTR} we may compute $\dim V_\xi$ when $\xi_0=1$ and $S$
intersects one of the coordinate planes $\{\xi_i=0\}_{i=1}^3$.  In
these intersections we obtain $\dim V_\xi=1$ except at the singular
point $\xi_{\operatorname{sing}}$ where $\dim V_\xi=2$.  \proofBox
\end{example}

\section{Determining the medium from the Fresnel surface}
\label{sec:Closure}

As described in the introduction, the new proof of implication
\ref{coIII} $\Rightarrow$ \ref{coII} in the next theorem is the first
main result of this paper.
Regarding the other implications let us make a few remarks.
Implication \ref{coII} $\Rightarrow$ \ref{coI} is a standard result
for the Hodge star operator on a $4$-manifold. The converse implication
\ref{coI} $\Rightarrow$ \ref{coII} is less well known. The result was
first derived by Sch\"onberg \cite{Shoenberg:1971}.  For further derivations and
discussions, see \cite{Jadczyk:1979, Rubilar2002, Obu:2003}.
Below we will give yet another proof using computer algebra. The proof
follows \cite{Obu:2003} and we use a Sch\"onberg-Urbantke-like formula
(see equations \eqref{eq:definitionOfGlobalG}--\eqref{eq:SU}) to
define a metric $g$ from $\kappa$.  However, the below argument that
$g$ transforms as a tensor seems to be new. For a different argument,
see \cite[Section D.5.4]{Obu:2003}.

When a general $2\choose 2$-tensor $\kappa$ on a $4$-manifold satisfies
$\kappa^2 = -f \operatorname{Id}$ as in condition \ref{coI} one says that
$\kappa$ satisfies the \emph{closure condition}. For physical
motivation, see \cite[Section D.3.1]{Obu:2003}. 
Let us emphasise that Theorem \ref{thm:mainResult} is a global result.
The result gives criteria for the existence of a Lorentz metric on a
$4$-manifold.  In general, we know that a connected manifold $M$ has a
Lorentz metric if and only if $M$ is non-compact, or if $M$ is compact
and the Euler number $\chi(N)$ is zero \cite[Theorem
2.4]{MingSan:2008}.  Let us also note that if $J$ is an \emph{almost
  complex structure} on a manifold $M$, that is, $J$ is a $1\choose
1$-tensor on $M$ with $J^2=-\operatorname{Id}$ and $\dim M\ge 2$, then
$M$ is orientable \cite[p. 77]{Hsiung:1995}. It does not seem to be
known if the analogous result also holds for $2\choose 2$-tensors, that is,
if the closure condition on a $4$-manifold implies orientability.

\begin{theorem}
\label{thm:mainResult}
Suppose $N$ is an orientable $4$-manifold.  If $\kappa\in
\Omega^2_2(N)$ satisfies $\kappaII=0$, then the following conditions
are equivalent:
\begin{enumerate}
\item $\kappa^2 = -f \operatorname{Id}$ for some function $f\in C^\infty(N)$ with $f>0$.
\label{coI}
\item 
  \label{coII} 
  There exists a Lorentz metric $g$ and a nonvanishing function $f\in
  C^\infty(N)$ such that
\begin{eqnarray}
   \kappa &=& f \ast_g.
\end{eqnarray}
\item 
\label{coIII} 
$\kappaIII=0$ and there exists a Lorentz metric $g$ such that 
\begin{eqnarray*}
  F(\kappa) &=& F(\ast_g),
\end{eqnarray*}
where $F(\kappa)$
is the Fresnel surface for $\kappa$ and $F(\ast_g)$
is the Fresnel surface for
the $2\choose 2$-tensor $\ast_g$.
\end{enumerate}
Moreover, when equivalence holds, then metrics $g$ in conditions \emph{\ref{coII}} and
\emph{\ref{coIII}} are conformally related.
\end{theorem}

\begin
{proof}
For implication \ref{coI} $\Rightarrow$ \ref{coII} let $\eta = f^{-1/2}
\kappa$ 
whence $\eta^2=-\operatorname{Id}$, and
let $h$ be an auxiliary positive definite Riemann metric on $N$. 
Let $\mathscr{T}$ be an atlas given by applying Lemma \ref{lemma:AinvLocally} 
to $\eta$. For the local claim, let $(U,x^i)$ be a chart in $\mathscr{T}$,
and in this chart let $\eta$ be
 represented by $3\times 3$ matrices $\cA$ and $\cK$.
With computer algebra we then obtain
\begin{eqnarray}
\label{eq:ThExp}
    \cG_{h,\eta}(\xi,\xi,\xi,\xi) &=& \operatorname{sgn} \det \cA\ (G^{ab} \xi_a \xi_b)^2, \quad \xi\in \Lambda^1(U), 
\end{eqnarray}
where $G=(G^{ab})$ is the $4\times 4$ matrix
\begin{eqnarray}
\label{eq:SU}
G &=& \frac 1 {(\det h)^{1/4}} \ \frac{1} {\vert \det \cA\vert^{1/2}} 
\left( \begin{array}{c|c}
  \det \cA &    k^i   \\
           \hline
k^j        & -\cA^{ij} + (\det \cA)^{-1} k^i k^j,
\end{array} \right),
\end{eqnarray}
and $k^i = \cA^{ib} \frac 1 2 \varepsilon_{bcd} \cK^{cd}$ for $i\in
\{1,2,3\}$.
Using a Shur complement \cite[Theorem 3.1.1]{Prasolov:1994} we find that
\begin{eqnarray*}
     \det G &=&   
\label{eq:Gdet}
-\frac{1}{\det h}.
\end{eqnarray*} 
Hence $\det G<0$, so matrix $G$ is invertible and has constant signature $(-+++)$ or
$(+---)$ in $U$.  Let $G_{ij}$ be the $ij$th entry of the inverse of
$G$. In $U$ we define
\begin{eqnarray}
\label{eq:definitionOfGlobalG}
  g &=& \sigma_U G_{ij} dx^i \otimes dx^j,
\end{eqnarray}
where constant $\sigma_U\in \{-1,1\}$ is chosen such that $g$ has signature
$(-+++)$.  Then $g$ defines a smooth symmetric $0\choose 2$-tensor 
in $U$ with signature $(-+++)$, and by computer algebra we have 
\begin{eqnarray}
\label{eq:localClaimEta}
\eta\vert_U &=& -\operatorname{sgn}  \operatorname{det} \cA \,  \ast_g.
\end{eqnarray}
This completes the local claim in \ref{coI} $\Rightarrow$ \ref{coII}.
For the global claim, let $(U, x^i)$ and $(\widetilde U,\widetilde
x^i)$ be overlapping charts in $\mathscr{T}$, and in these charts let
$G^{ij}$ and $\widetilde G^{ij}$ be defined as above. Since
$\cG_{h,\eta}$ is a tensor, equation \eqref{eq:ThExp} implies that
\begin{eqnarray} 
\label{eq:coordChange}
\operatorname{sgn} \det \cA\ (G^{ij} \xi_i \xi_j)^2
&=&
\operatorname{sgn} \det \widetilde \cA\ (\widetilde G^{ij} \pd{x^r}{\widetilde x^i}\pd{x^s}{\widetilde x^j}
\xi_r \xi_s)^2
\end{eqnarray}
for all $\xi=\xi_i dx^i\in \Lambda^1(U\cap \widetilde U)$. Since $G^{ab}$ is
non-degenerate we can find a $\xi$ such that the left hand side is
non-zero. Thus $\operatorname{sgn} \det \cA=\operatorname{sgn} \det
\widetilde \cA$ in $U\cap \widetilde U$ and $\operatorname{sgn} \det
\cA$ in equation \eqref{eq:localClaimEta} defines a smooth function
$N\to \setR$.  
By Theorem \ref{prop:Principle} \ref{eq:xx3} $\Rightarrow$ \ref{eq:xx1}
there exists a smooth nonvanishing function $\lambda\colon U\cap \widetilde U\to \setR$ such that 
\begin{eqnarray*}
 G^{ij} &=& \lambda \widetilde G^{rs} \pd{x^i}{\widetilde x^r}\pd{x^j}{\widetilde x^s}.
\end{eqnarray*}
Equation \eqref{eq:coordChange} implies that function $\lambda$ can only take values $\{-1,+1\}$.
Thus
\begin{eqnarray*}
 G_{ij} &=& \lambda \widetilde G_{rs} \pd{\widetilde x^r}{ x^i}\pd{\widetilde x^s}{ x^j}.
\end{eqnarray*}
Since $\sigma_U G_{ij}$ and $\sigma_{\widetilde U} \widetilde G_{ij}$ both have signature $(-+++)$.
It follows that $\lambda \sigma_U = \sigma_{\widetilde U}$ in $U\cap \widetilde U$, and equation
\eqref{eq:definitionOfGlobalG} defines a tensor on $N$.
This completes the proof of implication \ref{coI} $\Rightarrow$ \ref{coII}. 

Implication \ref{coII} $\Rightarrow$ \ref{coIII} 
follows by Propositions \ref{prop:HodgeHasOnlyPrincipalPart} and \ref{eq:ThgExp}. 

For the proof of implication \ref{coIII} $\Rightarrow$ \ref{coI} we first establish two subclaims:

\textbf{Claim $1$.}  The $4\choose 0$-tensor $\cG_{g,\ast g}$ is pointwise proportional to $\cG_{g,\kappa}$
by a non-zero constant.
 
Let $p\in N$. By Proposition \ref{prop:polarId} we only need to show that there exists  a $\lambda\in \setR$ such that
\begin{eqnarray*}
  \cG_{g,\ast g}(\xi,\xi,\xi,\xi) 
&=& \lambda \cG_{g,\kappa}(\xi,\xi,\xi,\xi), \quad \xi\in \Lambda^1_p(N).
\end{eqnarray*}
Let $x^i$ be coordinates around $p$ such that $g\vert_p=k
\operatorname{diag}(1,-1,-1,-1)$ for $k\in \{\pm 1\}$.  In these
coordinates, let $\cG^{ijkl}_{g,\ast g}$ and $\cG^{ijkl}_{g,\kappa}$
be components for the symmetric $4\choose 0$-tensors $\cG_{g,\ast
  g}\vert_p$ and $\cG_{g,\kappa}\vert_p$, so that
$$
  \cG_{g,\ast g}(\xi,\xi,\xi,\xi) =\cG^{ijkl}_{g,\ast_g} \xi_i \xi_j \xi_k \xi_l, \quad
\cG_{g,\kappa}(\xi,\xi,\xi,\xi) =\cG^{ijkl}_{g,\kappa} \xi_i \xi_j \xi_k \xi_l
$$
for $\xi = \xi_i dx^i\vert_p$. Using these components, let 
$P,Q$ be the polynomials $P,Q\colon \setR^4\to \setR$,
$$ 
   P(\xi_0,\vq)  = \cG^{ijkl}_{g,\ast_g} \xi_i \xi_j \xi_k \xi_l, \quad
  Q(\xi_0, \vq) =\cG^{ijkl}_{g,\kappa} \xi_i \xi_j \xi_k \xi_l,
$$ 
where $\xi_0\in \setR$, $\vq=(\xi_1,\xi_2,\xi_3) \in \setR^3$.
By
Proposition \ref{eq:ThgExp},
\begin{eqnarray*}
  P(\xi_0, \vq) &=& -(\xi_0^2 - \vert\vq\vert^2)^2\\
   &=& -\left(\xi_0-\vert\vq\vert \right)^2 \left(\xi_0+\vert\vq\vert \right)^2, 
\end{eqnarray*}
for all $(\xi_0, \vq)\in \setR^4$ when $\vert \vq\vert$ is the
Euclidean norm of $\vq$.  Thus $P(1,0,0,0)\neq 0$ so $dx^0\vert_p\notin
F_p(\ast_g)=F_p(\kappa)$ whence $\cG_{g,\kappa}^{0000}\neq 0$.  For each
$\vq\in \setR^3$, $Q(\xi_0,\vq)$ is then a fourth order polynomial in
$\xi_0$ with coefficients determined by $\vq\in \setR^3$. Hence there
exists continuous maps
$$
  r_i\colon \setR^3\to \setC, \quad i\in \{1,2,3,4\}
$$
so that for all $\vq\in \setR^3$, $\{r_i(\vq)\}_{i=1}^4$ are the roots
of $Q(\cdot, \vq)$ \cite{NaulinPabst:1994}. For each $\vq\in \setR^3$
there exists a $\alpha(\vq)\in \setR$ such that
\begin{eqnarray}
\label{eq:QxiPol}
 Q(\xi_0, \vq) &=& \alpha(\vq) \prod_{i=1}^4 (\xi_0 - r_i(\vq)), \quad \xi_0\in \setR.
\end{eqnarray}
Applying $\partial^4/\partial\xi_0^4$ to both sides implies that
$\alpha(\vq)=\cG^{0000}_{g,\kappa}$.  In particular, the map
$\vq\mapsto \alpha(\vq)$ is constant and non-zero. Let $\mu =
\cG^{0000}_{g,\kappa}$.
Since $P$ and $Q$ have the same zero set, there exists functions
$s_i\colon \setR^3\to \{-1,1\}$ such that
\begin{eqnarray*}
  r_i(\vq) &=& s_i(\vq) \vert \vq\vert, \quad \vq \in \setR^3, \, i\in \{1,2,3,4\}.
\end{eqnarray*}
We know that $\setR^3\slaz$ is path connected. Hence $\setR^3\slaz$ is
connected. For a contradiction, suppose that
$s_i(\setR^3\slaz)=\{-1,+1\}$ for some $i\in \{1,2,3,4\}$. Then
$\setR^3\slaz = U_+\cup U_-$ for open, non-empty and disjoint sets
$U_\pm$ defined as
\begin{eqnarray*}
  U_\pm &=& \{\vq\in \setR^3\slaz: \pm r_i(\vq) > 0\}.
\end{eqnarray*}
It follows that there are constants $s_1,s_2, s_3,s_4\in \{-1,+1\}$ such that 
\begin{eqnarray}
\label{eq:ridef}
  r_i(\vq) &=& s_i \vert \vq\vert, \quad \vq \in \setR^3, \, i\in \{1,2,3,4\}.
\end{eqnarray}
Let $\sigma$ be the number of $s_i$ with $s_i=1$.  If $\vq\in
\setR^3\slaz$, then polynomial $P(\cdot, \vq)$ has two distinct roots
$\pm \vert \vq\vert$.  Hence $\sigma=0$ or $\sigma = 4$ are not
possible, so $\sigma\in \{1,2,3\}$ and by equation \eqref{eq:QxiPol},
\begin{eqnarray*}
\label{eq:Qexp}
 Q(\xi_0, \vq) &=& \mu \left(\xi_0-\vert\vq\vert \right)^{\sigma} \left(\xi_0+\vert\vq\vert \right)^{4-\sigma}\\
  &=& \begin{cases}
\mu \left( \xi_0^4 - \vert \vq\vert^4 + 2\xi_0 \vert \vq\vert (\xi_0^2 -\vert\vq\vert^2)\right), &\mbox{if}\, \sigma=1,\\
\mu\, (\xi_0^2 - \vert \vq\vert^2)^2, &\mbox{if}\, \sigma=2,\\
\mu \left( \xi_0^4 - \vert \vq\vert^4 - 2\xi_0 \vert \vq\vert (\xi_0^2 -\vert\vq\vert^2)\right) , &\mbox{if}\, \sigma=3,
\end{cases}
\end{eqnarray*}
for all $(\xi_0, \vq)\in \setR^4$. Since $Q$ is a polynomial, we know
that $t\mapsto Q(1,t,0,0)$ is smooth near $0$. This is only possible when $\sigma=2$, and Claim $1$
follows.

\textbf{Claim $2$.} At each $p\in N$ there exists a non-zero
$\lambda\in \setR$ such that $\kappa\vert_p = \lambda \ast_g\vert_p$.

Let $p\in N$. By Proposition \ref{prop:localReduction} \ref{cc:CC}
there are coordinates $x^i$ around $p$ such that $g\vert_p=k
\operatorname{diag}(1,-1,-1,-1)$ for some $k\in \{\pm 1\}$ and
$\cA\vert_p$ is diagonal.  For $\xi = \xi_i dx^i\vert_p$ we then have
$$
   \cG_{g,\ast_g}(\xi,\xi,\xi,\xi) = \cG^{ijkl}_{g,\ast_g} \xi_i \xi_j \xi_k \xi_l, \quad
   \cG_{g,\kappa}(\xi,\xi,\xi,\xi) = \cG^{ijkl}_{g,\kappa} \xi_i \xi_j \xi_k \xi_l,
$$
where $\cG^{ijkl}_{g,\ast g}$ and $\cG^{ijkl}_{g,\kappa}$ are components for 
$\cG_{g,\ast_g}$ and $\cG_{g,\kappa}$ in coordinates $x^i$.
By Claim $1$ there exists a $\lambda\in \setR\slaz$ such that
\begin{eqnarray*}
\cG^{ijkl}_{g,\ast g} &=& \lambda \cG^{ijkl}_{g,\kappa}, \quad 0\le i\le j\le k\le l\le 3.
\end{eqnarray*}
Moreover, $\kappaIII=0$. We then have $36$ polynomial equations for
$\kappa$.  Using the Gr\"obner basis (see Appendix \ref{app:Groebner})
for these equations we find that the equations have a unique real
solution for $\kappa$ and this solution is given by $\kappa\vert_p =
\lambda^{-1/3} \ast_g\vert_p$.  This completes the proof of Claim $2$.

By Claim $2$, there exists a map $\lambda \colon N\to \setR\slaz$ such that
$\kappa = \lambda  \ast_g$ whence $\kappa^2 = -\lambda^2 \operatorname{Id}$.  To
see that $\lambda^2$ is smooth it suffices to note that $\lambda^2 = -\frac 1 6
\operatorname{trace} \kappa^2$.  This completes the proof of
implication \ref{coIII} $\Rightarrow$ \ref{coI}.


When equivalence holds, 
Proposition \ref{eq:ThgExp} and Theorem \ref{prop:Principle} imply that the
Lorentz metrics in conditions \ref{coII} and \ref{coIII} are conformally related.
\end{proof}


The next lemma was used to prove implication \ref{coI} $\Rightarrow$
\ref{coII} in Theorem \ref{thm:mainResult}.  In the proof of the lemma, Claim 1 
 is based on \cite[Sections D.4--D.5]{Obu:2003}.

\begin{lemma} 
\label{lemma:AinvLocally}
Suppose $N$ is an orientable $4$-manifold and $\kappa\in
\Omega^2_2(N)$.  If $\kappa$ has no skewon component and $\kappa^2 =
-\operatorname{Id}$, then $N$ has an oriented atlas $\mathscr{T}$ with
the following property:
Each $p\in N$ can be covered with a connected chart $(U,x^i)\in
\mathscr{T}$ such that if $\cA, \cB, \cC, \cD$ represent $\kappa$ in
$U$, then
\begin{enumerate}
\item $\cA$ is invertible in $U$.
\label{lemma:collectB}
\item In $U$ there exists a smoothly varying antisymmetric $3\times 3$ matrix $\cK$ such that
$$
\quad\quad
\cB=  -\cA^{-1} \left( \operatorname{Id}+ \left(\cK \cA^{-1}\right)^2\right),\quad
\cC = \cK\cA^{-1},\quad
\cD = -\cA^{-1} \cK.
$$
\label{lemma:collectD}
\end{enumerate}
\end{lemma}

\begin{proof}
  Let us first make an observations: Suppose $\{x^i\}_{i=0}^3$ are
  arbitrary coordinates for $N$ and $\cA$, $\cB$, $\cC, \cD$ are
  $3\times 3$ matrices that represent $\kappa$ in these
  coordinates. Then Proposition \ref{prop:localReduction} \ref{cc:AA}
  implies that $\kappa^2=-\operatorname{Id}$ is equivalent with
\begin{eqnarray}
\cC^2 + \cA\cB 
\label{eq:abc1} 
&=& -\operatorname{Id}, \\
\cB \cC + \cC^T \cB 
\label{eq:abc2}&=& 0, \\
\cC \cA + \cA\cC^T
\label{eq:abc3} 
&=& 0.
\end{eqnarray} 

Let $\mathscr{T}_0$ is a maximal oriented atlas for $N$. 
The proof is divided into two subclaims, Claim $1$ and Claim $2$.

\textbf{Claim 1.} For each $p\in N$ there exists a connected chart
$(U, x^i)$ that satisfy condition \ref{lemma:collectB} and there
exists a chart $(W, y^i)\in \mathscr{T}_0$ with $U\cap W\neq
\emptyset$ such that the transition map $x^i\mapsto y^i$ is
orientation preserving.

%
 By Proposition \ref{prop:localReduction} \ref{cc:BB} we can find a
 connected chart $(U, x^i)$ that contains $p$ and where matrix $\cA$ for
 $\kappa$ is diagonal at $p$.  The rest of Claim $1$ is divided into
 four cases depending on the eigenvalues of $\cA\vert_p$.

\textbf{Case A.} Suppose all three eigenvalues of $\cA\vert_p$ are
non-zero.  Since eigenvalues depend continuously on the matrix entries
\cite{NaulinPabst:1994}, we can shrink $U$ and 
part \ref{lemma:collectB} follows. Claim $1$ follows by possibly
reflecting the $x^1$-coordinate.

\textbf{Case B.} Suppose $\cA\vert_p$ has two non-zero eigenvalues. By
permutating the coordinates (see equation \eqref{eq:transRuleA}) we
may assume that $\cA\vert_p=\operatorname{diag}(a_1, a_2,0)$ for some
$a_1, a_2\neq 0$.  Writing out equations
\eqref{eq:abc1}--\eqref{eq:abc3} with computer algebra gives
$$
    \cC^1\,_1=   \cC^2\,_2=   \cC^3\,_1=\cC^3\,_2=0,\quad (\cC^3\,_3)^2=-1
$$
at $p$.  The last equation contradicts that $\cC$ is real. Case B is
therefore not possible.
 
\textbf{Case C.} Suppose $\cA\vert_p$ has one non-zero eigenvalue.  As
in Case B, we can find a chart $(U,x^i)$ for which $\cA\vert_p =
\operatorname{diag}(a_1,0,0)$ for some $a_1\neq 0$.  Writing out
equations \eqref{eq:abc1}--\eqref{eq:abc3} as in Case B gives
$$
   \cC^1\,_1=   \cC^2\,_1=   \cC^3\,_1=0,\quad \cB_{11}\neq 0, \quad \cC^2\,\,_3\neq 0,\quad \cC^3\,_2\neq 0
$$
at $p$. Let  $\{\widetilde x^i\}_{i=0}^3$ be coordinates around $p$ defined as
\begin{eqnarray*}
 \widetilde x^0 &=& x^0 + x^3,\\
\widetilde x^i &=& x^i, \quad i\in \{1,2,3\}.
\end{eqnarray*}
In these coordinates, matrix $\widetilde \cA\vert_p$ 
has determinant $-\cB_{11} (\cC^3\,_2)^2$, which is 
non-zero, and Claim $1$  follows as in Case $A$.
%

\textbf{Case D.} Suppose all eigenvalues of $\cA\vert_p$ are zero. Then
$\cA\vert_p=0$ and equation \eqref{eq:abc1} implies that
$(\det\cC\vert_p)^2=-1$. This contradicts that $\cC\vert_p$ is a real matrix. 
Case D is therefore not possible.

\textbf{Claim 2.} Let $\mathscr{T}$ be the collection of all charts
$(U, x^i)$ as in Step $1$ when $p$ ranges over all points in $N$. Then
$\mathscr{T}$ satisfies the sought properties.

Let $(U, x^i)$ and $(\widetilde U, \widetilde x^i)$ be overlapping
charts in $\mathscr{T}$. By Claim 1 and \cite[Lemma
13.9]{Trofimov:1994}, each chart $U$ and $\widetilde U$ is compatible
with all charts in $\mathscr{T_0}_0$. Hence the transition map $x^i
\mapsto \widetilde x^i$ is orientation preserving, and $\mathscr{T}$
is oriented.
We know that each chart in $\mathscr{T}$ satisfies property
\ref{lemma:collectB}, and property \ref{lemma:collectD} follows by
defining $\cK = \cC \cA$. Indeed, $\cK$ is antisymmetric by equation
\eqref{eq:abc3}, and the expression for $\cB$ follows by equation
\eqref{eq:abc1}.
\end{proof}

\section{Non-injectivity results}
\label{sec:uni}
Implication \ref{coIII} $\Rightarrow$ \ref{coII} in Theorem
\ref{thm:mainResult} shows that for a special class of medium, the
Fresnel surface determines the medium up to a conformal factor.  In
this section we will describe results and examples where the opposite
is true. In the below we will see that there are various
non-uniquenesses that prevents us from determining $\kappa$ (or even
the conformal class of $\kappa$) from only the Fresnel surface
$F(\kappa)$.

Let us study the non-injectivity of the two maps in the diagram below:
\begin{eqnarray}
\label{dia:diag2}
\kappa  \quad \to \quad \cG(\kappa) \quad \to \quad F(\kappa),\quad \quad\kappa\in \Omega^2_2(N),
\end{eqnarray}
where $\cG(\kappa)$ is the Tamm-Rubilar $4\choose 0$-tensor density
induced by $\kappa$.
%

\subsection{Non-injectivity of leftmost map}
\label{sec:TRinjective1}
Let us first study the non-injectivity of the leftmost map in diagram
\eqref{dia:diag2}, that is, the map
\begin{eqnarray}
\label{eq:LdefRmap}
\kappa\quad \to\quad \cG(\kappa), \quad\quad \kappa\in \Omega^2_2(N).
\end{eqnarray}
Parts \ref{thm:FkappaInvB:ii}--\ref{thm:FkappaInvB:iv} in the next
theorem describe three invariances that make the map in \eqref{eq:LdefRmap}
non-injective.  The first two parts are well known \cite[Section
2.2]{Obu:2003}.  However, let us make three remarks regarding part
\ref{thm:FkappaInvB:iv}.
First, an interpretation of  part \ref{thm:FkappaInvB:iv} is as follows:
If $F,G$ solve the sourceless Maxwell equations in medium $\kappa$,
then $G,F$ solve the sourceless Maxwell equations in medium $\kappa^{-1}$.
In this setting, part \ref{thm:FkappaInvB:iv} states that both media have the
same Fresnel surfaces.
Second, suppose $\ast_g$ is the $2\choose 2$-tensor induced by a
pseudo-Riemann metric $g$.  Then $\ast_g^2= \pm
\operatorname{Id}$, so $\ast_g^{-1}= \pm \ast_g$, whence
$\cG(\ast_g)$ and $\cG(\ast_g^{-1})$ are always conformally
related. Part \ref{thm:FkappaInvB:iv} states that this is
not only a result for Hodge star operators, but a general result
for \emph{all} $2\choose 2$-tensors.
Third, the proof of part \ref{thm:FkappaInvB:iv} is based on computer
algebra.  Of all the proofs in this paper, this computation is
algebraicly most involved.  For example, if we write out equation
\eqref{eq:bigEquation} as a text string, it requires almost $13$
megabytes of memory.

\begin{theorem}
\label{thm:FkappaInvB}
Suppose $\kappa\in \Omega^2_2(N)$ where $N$ is a $4$-manifold. Then
\begin{enumerate}
\item
\label{thm:FkappaInvB:i}
$\cG (f\kappa)=f^3 \cG(\kappa)$ for all $f\in C^\infty(N)$,
\item 
\label{thm:FkappaInvB:ii}
$\cG(\kappaII)=0$,
\item
\label{thm:FkappaInvB:iii}
$
  \cG(\kappa) = \cG(\kappa + f \operatorname{Id})
$ for all $f\in C^\infty(N)$,
\item
\label{thm:FkappaInvB:iv}
$\cG(\kappa^{-1})=\cG(-(\det\kappa)^{-1/3} \kappa)$ when
$\kappa$ is invertible.
\end{enumerate}
\end{theorem}

\begin{proof}
Part \ref{thm:FkappaInvB:i} follows by the definition, and parts
\ref{thm:FkappaInvB:ii}--\ref{thm:FkappaInvB:iii} are proven in \cite[Section 2.2]{Obu:2003}.
Therefore  we only need to prove part \ref{thm:FkappaInvB:iv}.
Let $\operatorname{adj} \kappa = \det \kappa \, \kappa^{-1}$ be the \emph{adjugate} of $\kappa$.
By part \ref{thm:FkappaInvB:i} it suffices to show that
\begin{eqnarray}
\label{eq:bigEquation}
(  \det \kappa)^2 \cG^{ijkl}_\kappa + \cG^{ijkl}_{\operatorname{adj} \kappa} &=& 0,
\quad 0\le i\le j\le k \le l \le 3,
\end{eqnarray}
where $\cG^{ijkl}_\kappa$ and $\cG^{ijkl}_{\operatorname{adj} \kappa}$
are components of the Tamm-Rubilar tensor densities of $\kappa$ and
$\operatorname{adj} \kappa$, respectively.  The motivation for rewriting
the claim as in equation \eqref{eq:bigEquation} is that now both terms are
polynomials. By using the method described in Appendix
\ref{app:veryLarge} we obtain that equations \eqref{eq:bigEquation}
hold, and part \ref{thm:FkappaInvB:iv} follows.
\end{proof}

Theorem \ref{thm:FkappaInvB}  \ref{thm:FkappaInvB:ii} 
shows that if we restrict the map in equation \eqref{eq:LdefRmap} to purely
skewon tensors, we do not obtain an injection.  The next example shows
that the same map is neither an injection when restricted to tensors of
purely principal type.

\begin{example}
On $N=\setR\times \setR^3$ with coordinates $\{x^i\}_{i=0}^3$, let $\kappa$ 
be the $2\choose 2$-tensor defined by $3\times 3$-matrices
\begin{eqnarray*}
\cA= 0_{3\times 3},\,\,\,
\cB= \begin{pmatrix} 
0 & 0  & \lambda_1 \\
0 & 0 & \lambda_2 \\
\lambda_1 & \lambda_2 & \lambda_3 
\end{pmatrix},\,\,\,
\cC= \begin{pmatrix}
-2^{-1/3} & 0 & \lambda_4 \\
0 & -2^{-1/3} & \lambda_5 \\
0 & 0& 2^{2/3}
\end{pmatrix},\,\,\,
\cD = \cC^T,
\end{eqnarray*}
where parameters $\lambda_1,\ldots, \lambda_5\in \setR$ are arbitrary.
Then $\kappa$ has only a principal part, $\det \kappa = 1$, and 
\begin{eqnarray*}
\cG_{h,\kappa}(\xi,\xi,\xi,\xi) &=& 0, \quad \xi\in \Lambda^1(N)
\end{eqnarray*}
for any pseudo-Riemann metric $h$ on $N$.
\proofBox
\end{example}

When proving implication \ref{coIII} $\Rightarrow$ \ref{coI} in
Theorem \ref{thm:mainResult} we need to assume that $\kappa$ has real
coefficients.  In fact, for $2\choose 2$-tensors with complex
coefficients a decomposition into principal-, skewon-, and axion
components does not seem to have been developed. The next example
shows that there are non-trivial complex tensors whose Fresnel surface
everywhere coincides with the Fresnel surface for the standard
Minkowski metric $g_0 = \operatorname{diag}(1,-1,-1,-1)$.  (For
$\kappa\in \Omega^2_2(N,\setC)$ we define the Fresnel surface using
the same formulas as for real $\kappa$.) 

\begin{example}
\label{ex:complexExample}
On $N=\setR \times \setR^3$ with coordinates $\{x^i\}_{i=0}^3$, let $\kappa$ 
be the $2\choose 2$-tensor with complex coefficients defined by $3\times 3$-matrices 
\begin{eqnarray*}
\cA&=& 
-\begin{pmatrix}
\frac{1}{2 z^2} & 0 & 0 \\
0 & 2z & 0 \\
0 & 0& z
\end{pmatrix},\,\,\,
\quad\quad\quad\quad\quad\quad\quad\,\,
\cB= -\cA, \quad\\
\cC&=& 
i \begin{pmatrix}
\frac 1 {3z^2} -z & 0 & 0 \\
0 & -\frac {1}{6z^2} + z & 0 \\
0 & 0& -\frac{1}{6z^2}
\end{pmatrix},\,\,\,\quad\quad
\cD = \cC,
\end{eqnarray*}
where $z$ is an arbitrary function $z\colon N \to \setC\slaz$ and $i$ is the complex unit.
At each $p\in N$ the Fresnel surface is then determined by
\begin{eqnarray*}
 \xi_0^2 - \xi_1^2-\xi_2^2-\xi_3^2&=&0, 
\end{eqnarray*}
where $\xi_i dx^i \in \Lambda^1_p(N)$, and
\begin{eqnarray*}
\operatorname{trace} \kappa = 0, \quad
 \det \kappa = \frac{\left(1+6z^3\right)^3 \, \left( 5 - 126 z^3 +684 z^6 - 648 z^9 \right) }{ 46 656\,\, z^{12}}.
\end{eqnarray*}
From the latter equation we see that for specific values of $z$, tensor $\kappa$ can be non-invertible as a linear map
\proofBox
\end{example}

\subsection{Non-injectivity of rightmost map}
\label{sec:TRinjective2}
The next example shows that there are skewon-free $2\choose 2$-tensors
$\kappa_1$ and $\kappa_2$ that have the same Fresnel surfaces, but
their Tamm-Rubilar tensors are not proportional to each other. This
shows that the rightmost map in equation \eqref{dia:diag2} is not
injective.  Let us point out that this contradicts the first
proposition in \cite{PunziEtAl:2009} whose proof does not analyse
multiplicities of roots to the Fresnel equation.

\begin{example}
On $N=\setR\times \setR^3$ with coordinates $\{x^i\}_{i=0}^3$, let $\kappa_1$ 
be the $2\choose 2$-tensor defined by $3\times 3$-matrices
\begin{eqnarray*}
\cA_1= \begin{pmatrix}
0 & -1 & 1 \\
-1 & -2 & 1 \\
1 & 1 & -1 
\end{pmatrix},\,\,\,
\cB_1= \begin{pmatrix}
0 & \frac 1 2  & 0 \\
\frac 1 2 & 0 & 0 \\
0 & 0 & 0 
\end{pmatrix},\,\,\,
\cC_1= \begin{pmatrix}
0 & 0 & 0 \\
0 & 2 & 1 \\
\frac 1 2 & -\frac 1 2 & 1 
\end{pmatrix},\,\,\,
\cD_1 = \cC_1^T.
\end{eqnarray*}
For the Euclidean metric $g_0$ on $N$ we then have
\begin{eqnarray*}
\cG_{g_0,\kappa_1}(\xi,\xi,\xi,\xi) &=& (\xi_0-\xi_1)(\xi_0-\xi_2)^3, \quad \xi\in \Lambda^1(N).
\end{eqnarray*}
To exchange the role of $\xi_1$ and $\xi_2$, we perform a coordinate
change $x_0\mapsto x_0$, $x_1\mapsto x_2$, $x_2\mapsto x_1$,
$x_3\mapsto x_3$.  With this as motivation we define $\kappa_2$ as the
$2\choose 2$-tensor defined by $3\times 3$-matrices
\begin{eqnarray*}
  \cA_2= \begin{pmatrix}
2 & 1  &-1 \\
1  & 0 & -1 \\
-1 & -1 & 1 
\end{pmatrix},\,\,\,
 \cB_2= \begin{pmatrix}
0 & -\frac 1 2  & 0 \\
-\frac 1 2 & 0 & 0 \\
0 & 0 & 0 
\end{pmatrix},\,\,\,
 \cC_2= \begin{pmatrix}
2 & 0 & 1 \\
0 & 0 & 0 \\
-\frac 1 2 & \frac 1 2 & 1 
\end{pmatrix},\,\,\,
\cD_2 = \cC_2^T.
\end{eqnarray*}
Then
\begin{eqnarray*}
\cG_{g_0,\kappa_2}(\xi,\xi,\xi,\xi) &=& -(\xi_0-\xi_1)^3(\xi_0-\xi_2), \quad \xi\in \Lambda^1(N).
\end{eqnarray*}
Here $\kappa_1$ and $\kappa_2$ are not proportional, 
their Tamm-Rubilar tensor densities are not proportional, but their Fresnel surfaces coincide.

Both $\kappa_1$ and $\kappa_2$ have $1$ has an eigenvalue of algebraic multiplicity $6$.
Hence
$$
  \det \kappa_1 = \det \kappa_2 = 1, \quad \operatorname{trace}\kappa_1 = \operatorname{trace}\kappa_2 = 6,
$$
and for the trace-free components $\widetilde \kappa_i = \kappa_i-\operatorname{Id}$ we have $\det \widetilde \kappa_i=0$.
\proofBox
\end{example}

\textbf{Acknowledgements.}  I would like to thank Luzi Bergamin and
Alberto Favaro for helpful discussions by email regarding
\cite{FavaroBergamin:2011}.

The author gratefully appreciates financial
support by the Academy of Finland (project 13132527 and Centre of
Excellence in Inverse Problems Research), and by the Institute of
Mathematics at Aalto University.

\appendix
\section{Gr\"obner bases}
\label{app:Groebner}
To prove implication \ref{coIII} $\Rightarrow$ \ref{coI} in Theorem
\ref{thm:mainResult} 
we use a Gr\"obner basis to solve a system of
polynomial equations. In this appendix we collect the results about
Gr\"obner bases that are needed for this step in the proof. These
results are standard and can, for example, be found in \cite[
pp. 5, 
29--32, 
76--77] {CoxLittleOShea:2007}. 

By $\setC[x]$ we denote the ring of polynomials with complex
coefficient that depend on variables $x_1, \ldots, x_N\in \setC$ where
$N\ge 1$. A \emph{(polynomial) ideal} is a subset $I\subset \setC[x]$
such that
\begin{enumerate}
\item $0\in I$,
\item $f,g\in I$ implies that $f+g \in I$,
\item $f\in I$ and $g\in \setC[x]$ implies that $fg\in I$.
\end{enumerate}

\begin{theorem}[Hilbert basis theorem]
\label{thm:HBT}
If $I\subset \setC[x]$ is an ideal, then there exists 
finitely many polynomials $f_1, \ldots, f_s\in I$ such that 
\begin{eqnarray}
\label{eq:explicit}
I &=& \left\{ \sum_{i=1}^s f_i g_i :  g_1, \ldots, g_s \in \setC[x]\, \right\}.
\end{eqnarray} 
\end{theorem}

When Theorem \ref{thm:HBT} holds, we say that polynomials $f_1,
\ldots, f_s$ form a \emph{basis} for ideal $I$. Conversely, if $f_1,\ldots, f_s$ are
any polynomials in $\setC[x]$, the ideal on the right hand side in equation
\eqref{eq:explicit} is denoted by $\langle f_1,\ldots, f_s\rangle$ and called
the \emph{ideal generated by polynomials} $f_1, \ldots, f_s$.
The \emph{affine variety}  defined by polynomials $f_1, \ldots, f_s\in \setC[x]$ is 
the subset $V(f_1, \ldots, f_s)\subset \setC^N$,
\begin{eqnarray*}
  V(f_1, \ldots, f_s) &=& \{ x\in \setC^N : f_1(x)=\cdots=f_s(x)=0\,\,\}.
\end{eqnarray*}

The next proposition gives sufficient conditions for two systems of
polynomial equations to have the same solution sets.

\begin{proposition} 
\label{prop:GrBaLemma}
Suppose $f_1, \ldots, f_s$ and $g_1, \ldots, g_t$ are polynomials in $\setC[x]$.
If $\{f_i\}$ and $\{g_i\}$  generate the same ideal, that is, 
\begin{eqnarray*}
  \langle f_1,\ldots, f_s \rangle &=& \langle g_1, \ldots, g_t\rangle,
\end{eqnarray*}
then
\begin{eqnarray*}
  V(f_1, \ldots, f_s) &=& V(g_1, \ldots, g_t).
\end{eqnarray*}
\end{proposition}  

%

We will not give a precise definition for a Gr\"obner basis.  The key
property for  Gr\"obner bases is collected in the next proposition.

\begin{proposition} 
\label{prop:GrBa}
Let $f_1, \ldots, f_s\in \setC[x]$ be polynomials such that $\langle
f_1, \ldots, f_s\rangle \neq \{0\}$. Then there exists polynomials
$g_1, \ldots, g_t\in \setC[x]$ such that
\begin{eqnarray*}
  \langle f_1, \ldots, f_s \rangle &=&   \langle g_1, \ldots, g_t \rangle.
\end{eqnarray*}
Polynomials $g_i$ are called a \emph{Gr\"obner basis} for the ideal $\langle f_1, \ldots, f_s\rangle$.
\end{proposition}   

Even if computation of Gr\"obner basis computation is supported by modern
computer algebra systems, their computation can in practice be very time
consuming. The motivation for using a Gr\"obner bases is that they
typically simplify the solution process for polynomial equations. Thus
one can think of Gr\"obner bases as a way to simplify polynomial
equations without changing their solution set.  This is illustrated in
the next example.

\begin{example} 
Let $S\subset \setR^3$ be all $(x,y,z)\in \setR^3$ such that 
\begin{eqnarray}
\label{eq:originalPolEq}
 x y z  = 1, \quad 
 x z^2= y^2, \quad 
 z^2 =x y.
\end{eqnarray}
By elementary manipulation, we see that $S=\{(1,1,1)\}$. To
illustrate how to determine $S$ using a Gr\"obner a basis, let us first note that 
$S=V(f_1,f_2,f_3)\cap \setR^3$, where
$$
  f_1 = x y z -1, \quad f_2 = x z^2-y^2, \quad f_3 = z^2-xy.
$$
With computer algebra we find that a Gr\"obner basis for $\langle f_1, f_2, f_3\rangle$ is given by
$$
 g_1 = z^3-1, \quad g_2 = y^3-z, \quad g_3 = x-y^2 z.
$$
Propositions \ref{prop:GrBaLemma} and \ref{prop:GrBa} imply that $
V(f_1,f_2,f_3) = V(g_1, g_2, g_3)$. Hence $S$ coincide with real
solutions to polynomial equations
\begin{eqnarray}
\label{eq:SaltEq}
  z^3=1, \quad y^3=z, \quad x=y^2 z.
\end{eqnarray}
These last equations are easily solved and we find that $S=\{(1,1,1)\}$.

If we compare the original equations \eqref{eq:originalPolEq} to 
equations \eqref{eq:SaltEq} computed by a Gr\"obner basis, we see 
that the latter ones are more easier to  solve since they can be solved by
backsubstitution.  \proofBox
\end{example}

\section{Verifying very large polynomial identities}
\label{app:veryLarge}

The proof of Theorem \ref{thm:FkappaInvB} \ref{thm:FkappaInvB:iv}
reduces to proving equations \eqref{eq:bigEquation} which consists of
$35$ polynomial identities in $36$ variables. If we write out these
polynomial identities as text strings, they occupy almost $13$
megabytes of memory. Due to this size, Mathematica (version 7.0.1) was
not able to verify the identities in a reasonable time.  In this
appendix we describe a recursive method that is able to verify these
identities.  On a computer with two Intel E8400 $3$GHz processors and
$3.7$ gigabytes of RAM the method finished in $10$ hours. 
The method relies on the following corollary to Taylor's theorem with
a Lagrange error term.

\begin{proposition} 
\label{prop:reducePoly} 
Suppose $f$ is a polynomial $f\colon \setR^N\to \setR$ 
in variables $x_1, \ldots, x_N \in \setR$. Furthermore, suppose that
\begin{enumerate}
\item \label{tay:coI}
There exists a finite $K\in \{1,2,\ldots\}$ such that 
\begin{eqnarray*}
 \frac{\partial^K f}{\partial x_1^K}(x) &=& 0 \quad \mbox{for all}\,\, x\in \setR^N.
\end{eqnarray*}
\item Polynomials $Z^0, \ldots, Z^{K-1}\colon \setR^{N-1}\to \setR$ defined as
\label{tay:coII}
\begin{eqnarray*}
 Z^{r}(y) &=& \frac{\partial^{r} f}{\partial x_1^{r}}(0,y), \quad \, y\in \setR^{N-1}, \quad r=0, \ldots, K-1
\end{eqnarray*}
are zero polynomials.
\end{enumerate}
Then $f$ is the zero polynomial.
\end{proposition}


Proposition \ref{prop:reducePoly} shows that to verify identity $f=0$
we only need to verify identities $Z^0=0$, $\ldots$,
$Z^{K-1}=0$. Since these identities are obtained by
differentiating $f$ and by setting one variable to zero, they are
typically shorter and easier to manipulate than the original $f$. By
recursively applying Proposition \ref{prop:reducePoly}, the proof of
$f=0$ divides into smaller and smaller polynomial identities that
eventually can be verified using Mathematica's internal Expand
routine.
%
%
The implementation details are as follows. To verify $f=0$ we applied
Proposition \ref{prop:reducePoly} recursively until the polynomial had
less than $27$ variables (out of $36$ original).  For each application of
Proposition \ref{prop:reducePoly} we used the (non-optimal) constant 
$K=5$.

\providecommand{\bysame}{\leavevmode\hbox to3em{\hrulefill}\thinspace}
\providecommand{\MR}{\relax\ifhmode\unskip\space\fi MR }
\providecommand{\MRhref}[2]{%
  \href{http://www.ams.org/mathscinet-getitem?mr=#1}{#2}
}
\providecommand{\href}[2]{#2}

\end{document}